\documentclass[11pt,reqno]{amsart}
\usepackage{amssymb}
\usepackage{amsfonts}
\usepackage{amsmath}
\usepackage{stmaryrd}
\usepackage{physics}
\usepackage{braket}
\usepackage{dsfont}
\usepackage{bbold}
\usepackage{graphicx}
\usepackage{relsize}
\usepackage[table,xcdraw]{xcolor}
\usepackage{enumerate}
\usepackage{hyperref}
\hypersetup{colorlinks = true, linkcolor = blue, urlcolor  = blue, citecolor = red}
\usepackage[margin=1in]{geometry}
\usepackage{enumitem}
\usepackage{mathtools}
\usepackage{graphbox}
\usepackage{float}
\usepackage[font=scriptsize]{caption}
\usepackage{MnSymbol}
\usepackage{appendix}

\renewcommand{\epsilon}{\varepsilon}
\renewcommand{\phi}{\varphi}
\newcommand{\overbar}[1]{\mkern 1.5mu\overline{\mkern-1.5mu#1\mkern-1.5mu}\mkern 1.5mu}
\newtheorem{theorem}{Theorem}[section]
\newtheorem{definition}[theorem]{Definition}
\newtheorem*{theorem*}{Theorem}
\newtheorem*{definition*}{Definition}
\newtheorem{proposition}[theorem]{Proposition}
\newtheorem{corollary}[theorem]{Corollary}
\newtheorem{lemma}[theorem]{Lemma}
\newtheorem{remark}[theorem]{Remark}
\newtheorem{example}[theorem]{Example}

\newtheorem*{conjecture*}{Conjecture}

\DeclareMathOperator{\Cf}{Cf}

\definecolor{darkgreen}{rgb}{0,0.392,0}

\author{Ion Nechita}

\address{Laboratoire de Physique Th\'eorique, Universit\'e de Toulouse, CNRS, UPS, France}
\email{nechita@irsamc.ups-tlse.fr}

\author{Satvik Singh}
\email{satviksingh2@gmail.com}
\address{\parbox{\linewidth}{Department of Physical Sciences,\\
Indian Institute of Science Education and Research (IISER) Mohali, Punjab, India.}}

\title[A graphical calculus for integration over random diagonal unitary matrices]{A graphical calculus for integration over\\random diagonal unitary matrices}

\begin{document}

\begin{abstract}
    We provide a graphical calculus for computing averages of tensor network diagrams with respect to the distribution of random vectors containing independent uniform complex phases. Our method exploits the order structure of the partially ordered set of uniform block permutations. A similar calculus is developed for random vectors consisting of independent uniform signs, based on the combinatorics of the partially ordered set of even partitions. We employ our method to extend some of the results by Johnston and MacLean on the family of local diagonal unitary invariant matrices. Furthermore, our graphical approach applies just as well to the real (orthogonal) case, where we introduce the notion of triplewise complete positivity to study the condition for separability of the relevant bipartite matrices. Finally, we analyze the twirling of linear maps between matrix algebras by independent diagonal unitary matrices, showcasing another application of our method.
    
\end{abstract}



\maketitle

\tableofcontents

\section{Introduction}

Since the early days of random matrix theory, diagrammatic methods have played an important role in its development, both as a proof technique (mainly used together with the method of moments) and as a bridge towards combinatorial topics, such as map enumeration \cite{zvonkin1997matrix,eynard2016counting}. Proliferation of modern fields such as data science and quantum information theory---which focus on multi-linear generalizations of the classical matrix theory---has only propelled the study of random tensors, both from a theoretical and practical standpoint. Representing complicated tensor networks using diagrams containing boxes and wires is an old idea going back to Penrose \cite{penrose1971applications}. Since then, graphical techniques have found widespread applications ranging from simulations of quantum many body systems \cite{Vidal2008manybodysimulations, Verstraete2008renormalization} to development of tensor networks for open quantum systems and quantum circuits \cite{wood2015tensor, Bergholm2011circuit}. Recently, such representations were recognized to be of relevance in condensed matter theory, where low-energy quantum states admit concise descriptions in terms of tensor networks \cite{orus2014practical}. 

Tensor networks containing \emph{random} objects have been analyzed in connection to condensed matter physics \cite{collins2013matrix} and holographic duality \cite{hayden2016holographic}. Fully graphical calculi for computing averages of such random tensor networks have been introduced in the mathematical physics literature in order to simplify the often cumbersome algebraic solutions to such problems: initially for Haar-distributed random unitary matrices \cite{collins2010random}, and then for Gaussian tensors \cite{collins2011gaussianization}. These methods have been used to great success in quantum information theory \cite{collins2016random}, due to a central place occupied by tensor product constructions in the theory. These calculi have been recently implemented for computer algebra systems \cite{fukuda2019rtni}.

The main contribution of this work is a graphical calculus for computing averages of tensor network diagrams containing random vectors having independent and identically distributed (i.i.d.) uniformly random phases as their entries. Given several copies of such random vectors $u \in \mathbb C^d$ (and of its complex conjugate $\bar u$) in a diagram $\mathcal D$, our main result allows one to express the expectation value of $\mathcal D$ with respect to the distribution of $u$ as a \emph{weighted sum of diagrams} constructed from $\mathcal D$ by connecting together the half-edges incident to the $u$ (and $\overbar u$) vectors. We refer the reader to Theorem \ref{theorem:E-u} for the precise statement and to Section \ref{sec:combinatorics} for the combinatorial background. 

\begin{theorem*}
Given a tensor network $\mathcal D$ containing one or more copies of a random vector $u$ with i.i.d. uniformly random phases, we have 
$$\mathbb{E}_u \mathcal{D} = \sum_{(\alpha, \beta, f) \in \mathcal{UBP}} \mathcal{D}_{(\alpha, \beta, f)} \Cf_{\mathcal U}(\alpha, \beta, f),$$
where the sum is indexed by \emph{uniform block permutations} and $\Cf_{\mathcal U}$ are the combinatorial weights.
\end{theorem*}

Once it is identified that the pairings of different half-edges associated with the $u$ and $\overbar u$ boxes will be effected by \emph{uniform block permutations} \cite{aguiar2008hopf, maurice-ubp}, the crucial idea in the proof of the theorem is to invoke M\"obius inversion in order to utilize the partial order on the set of uniform block permutations \cite{stanley-combinatorics-1}. We compute explicitly the coefficients $\Cf_{\mathcal U}$ which have a multiplicative structure, fitting nicely within the framework of Stanley's \emph{exponential structures} \cite{stanley1978exponential}. An almost identical result is obtained in the real case, where random phases are replaced by uniform random signs and the partially ordered set of \emph{even partitions} becomes the relevant combinatorial object. The upshot of our results is that the combinatorial weights ensure that no double-counting of diagrams occurs, a major inconvenience of the direct, algebraic (or index-based) approach. 

We then present several applications of this result, mainly to the theory of quantum information. We focus on the notion of (conjugate) \emph{local diagonal unitary invariant} matrices (abbreviated (C)LDUI) introduced in \cite{chruscinski2006class} and further studied by Jonhston and MacLean in \cite{johnston2019pairwise} within the context of the \emph{absolutely separable vs.~absolutely PPT conjecture} \cite{IQOQIseparability}. We analyze in detail the real case of \emph{local diagonal orthogonal invariant} (LDOI) matrices, introducing the notion of triplewise completely positive matrices. The cone of LDOI matrices is of great interest in quantum information theory, as many of the most important families of PPT entangled states belong in this class \cite{chruscinski2006class}. The structural properties of these quantum states and the corresponding covariant quantum channels are analyzed in detail in a subsequent work \cite{singh2020diagonal,singh2020ppt2}. Building upon the techniques used in this paper, a novel graph theoretic protocol for detecting entanglement in arbitrary bipartite states has been developed in \cite{Singh2020entanglement}. 

The paper is organized as follows. Section \ref{sec:combinatorics} contains the necessary combinatorial background on the theory of partially ordered sets and M\"obius inversion, focusing on the cases of uniform block permutations and even partitions. Section \ref{sec:diagrams-tensors} gives a concise introduction to the graphical representation of tensors, the language in which the main results of this paper are stated. Sections \ref{sec:graphical-integration-C} and \ref{sec:graphical-integration-R} are the core of the paper, containing the main results, Theorems \ref{theorem:E-u} and \ref{theorem:E-s}. The final Sections \ref{sec:LDUI}-\ref{sec:twirl} contain different applications of these results: local diagonal unitary invariant matrices, local diagonal orthogonal invariant matrices, and diagonal twirling of linear maps between matrix algebras respectively. Finally, we gather in Appendix \ref{sec:app-n-3} the 16 diagrams appearing when averaging a diagram consisting of three $u$ vectors and three $\bar u$ vectors, and in Appendix \ref{sec:app-tcp} some further properties of triplewise completely positive matrices, generalizing the case of pairwise completely positive matrices from \cite{johnston2019pairwise}. 

\section{Combinatorial prerequisites}\label{sec:combinatorics}

We gather in this section some basic definitions and facts from combinatorics, mainly from the theory of set partitions, uniform block permutations and M\"obius functions.
\subsection*{Set partitions}
\begin{definition} \label{def:part}
A partition $\alpha$ of a non-empty set $X$ is a collection of non-empty, pairwise disjoint subsets of $X$ such that their union equals $X$. 
\end{definition}
The set of all partitions of $X$ is denoted by $\Pi_X$. For $n \in \mathbb{N}$, the set of all partitions of $[n]=\{1,2,\ldots, n\}$ is denoted by $\Pi_n$. $\vert X \vert$ denotes the cardinality of the set $X$. We represent the partition $\displaystyle\{ \{1,2\}, \{3\}, \{4\} \} \in \Pi_4$ as $12\vert 3\vert 4$. The type $\lambda \coloneqq 1^{m_1}2^{m_2}\ldots n^{m_n}$ of a partition $\alpha \in \Pi_n$ is the integer partition of $n$ (denoted by $\lambda \vdash n$) formed by the block sizes of $\alpha$: $\sum_{i=1}^n im_i = n$.

\begin{definition}\label{def:ordpart}
An ordered partition $\vec{\alpha}$ of a non-empty set $X$ is a vector of non-empty, pairwise disjoint subsets of $X$ such that their union equals $X$.
\end{definition} 
For a non empty set $X$ and $n\in \mathbb{N}$, the ordered counterparts of $\Pi_X$ and $\Pi_n$ are denoted by $\vec{\Pi}_X$ and $\vec{\Pi}_n$ respectively. As is evident from the definition, for an ordered partition $\vec{\alpha}$, the ordering of blocks in $\vec{\alpha}$ is important but the ordering of elements within the blocks is not. On the other hand, for an ordinary partition $\alpha$, neither of them matter. Hence, for $\lambda \vdash n$, the number of partitions of type $\lambda$ in $\Pi_n$ is 
\begin{equation} \label{eq:no-part}
P_n(\lambda) = \frac{n!}{m_1!m_2!\ldots m_n!(1!)^{m_1}(2!)^{m_2}\ldots (n!)^{m_n}}
\end{equation}
while the number of ordered partitions of type $\lambda$ in $\vec{\Pi}_n$ is
\begin{equation} \label{eq:no-o-part}
O_n(\lambda) = \frac{n!}{(1!)^{m_1}(2!)^{m_2}\ldots (n!)^{m_n}}
\end{equation} 
Summing the above expressions over all integer partitions gives the cardinalities of $\Pi_n$ and $\vec{\Pi_n}$:
\begin{equation}
    \vert \Pi_n \vert = \sum_{\lambda \vdash n} P_n(\lambda) \qquad \vert \vec{\Pi_n} \vert = \sum_{\lambda \vdash n} O_n(\lambda)
\end{equation}

\begin{example} 
For n=2, we have $\Pi_2 = \{12, 1\vert 2\}$ and $\vec{\Pi}_2 = \{12, 1\vert 2, 2\vert 1\}$.
\end{example}

\begin{definition}
An even partition $\alpha$ of a non-empty set $X$ with even cardinality is a collection of non-empty, pairwise disjoint subsets of $X$ with even cardinalities such that their union equals $X$. 
\end{definition}
For $n\in \mathbb{N}$, the set of all even partitions of $[2n]$ is denoted by $\Pi^{(2)}_{2n}$. Partitions in $\Pi^{(2)}_{2n}$ are of type $\lambda = 2^{a_1} 4^{a_2}\ldots (2n)^{a_n}$, which are nothing but even integer partitions of $2n$ (denoted by $\lambda \Vdash 2n$): $\sum_{i=1}^n 2i a_i = 2n$. From Eq.~\eqref{eq:no-part}, the cardinality of $\Pi^{(2)}_{2n}$ is readily obtained
\begin{equation}
\vert \Pi^{(2)}_{2n} \vert = \sum_{\lambda \Vdash 2n} P_{2n}(\lambda)
\end{equation}
The initial values of the sequence $\{ \vert \Pi^{(2)}_{2n} \vert \} _{n\in \mathbb{N}}$ are given below (\cite{oeis} \href{https://oeis.org/A005046}{A005046}). We will later construe these values as the total number of diagrams present in a certain graphical integration formula.
\begin{equation}
    1,4,31,379,6556,150349,\ldots
\end{equation}
\subsection*{Uniform block permutations}

\begin{definition}[Uniform Block Permutations]\label{def:1}
A uniform block permutation (UBP) of $[n]$ is a triple $(\alpha, \beta, f)$ such that 
\begin{enumerate}
\item $\alpha, \beta$ are partitions of $[n]$ of the same type $\lambda \vdash n$.
\item $f:\alpha \rightarrow \beta$ is a bijection such that $\vert f(a) \vert = \vert a \vert \, \forall a\in \alpha$.
\end{enumerate} 
\end{definition}

The set of all UBPs of $[n]$ is denoted by $\mathcal{UBP}_n$. Given $\alpha, \beta \in \Pi_n$ of the same type $\lambda$, the bijection $f:\alpha \rightarrow \beta$ can be thought of as providing an ordering of the blocks of $\beta$. This is stated more precisely in the following lemma

\begin{lemma} \label{lemma:ubp}
The set of uniform block permutations of $[n]$ is in bijection with the set of pairs $(\alpha, \vec{\beta})$, where $\alpha \in \Pi_n, \vec \beta \in \vec{\Pi}_n$ are of the same type $\lambda \vdash n$. 
\end{lemma}

For $\alpha = \alpha_1\vert \alpha_2\vert \ldots\vert \alpha_k \in \Pi_n$, we represent $(\alpha, \beta, f) \in \mathcal{UBP}_n$ in the form of an array with two rows, where the blocks of $\beta$ are ordered so that $\beta_i = f(\alpha_i)$ for $i=1,2,\ldots,k$.
\begin{equation*}
\left(\!\!\begin{tabular}{c|c|c|c}
        $\alpha_1$ & $\alpha_2$ & $\ldots$ & $\alpha_k$ \\
        $\beta_1$ & $\beta_2$ & $\ldots$ & $\beta_k$
    \end{tabular}\!\!\right)
\end{equation*}

From Lemma \ref{lemma:ubp}, it is clear that the number of UBPs of type $\lambda \vdash n$ in $\mathcal{UBP}_n$ is given by 
\begin{equation}\label{eq:num-ubp}
U_n(\lambda) = P_n(\lambda)O_n(\lambda) = \left( \frac{n!}{(1!)^{m_1}(2!)^{m_2}\ldots (n!)^{m_n}} \right)^2 \frac{1}{m_1!m_2!\ldots m_n!}
\end{equation}
The cardinality of $\mathcal{UBP}_n$ is obtained by summing the above expression over all integer partitions. 
\begin{equation}
    \vert \mathcal{UBP}_n \vert = \sum_{\lambda \vdash n} U_n(\lambda)
\end{equation}
The first few values of the sequence $\{ \vert \mathcal{UBP}_n \vert \}_{n\in \mathbb{N}} $ are stated below for convenience (\cite{oeis} \href{https://oeis.org/A023998}{A023998}). These will later be interpreted as the total number of diagrams present in some graphical integration formula. 
\begin{equation}
    1,3,16,131,1496,22482,\ldots
\end{equation}

We now wish to impart the structure of an inverse semigroup on $\mathcal{UBP}_n$.
\begin{definition} \cite{maurice-ubp}
For $(\alpha,\beta, f), (\alpha', \beta', g) \in \mathcal{UBP}_n$, the composition $(\gamma, \delta, g\circ f) \coloneqq (\alpha', \beta', g) \circ (\alpha,\beta, f)$ is defined such that the blocks $\gamma_i$ of $\gamma$ are minimal with respect to the following properties
\begin{enumerate}
\item $\gamma_i$ is a union of blocks $\alpha_j$ of $\alpha$.
\item $f(\gamma_i)$ is a union of blocks $\alpha'_j$ of $\alpha'$.
\end{enumerate}
The images $g\circ f(\gamma_i)$ are then union of the images $g(\alpha'_j)$.
\end{definition}

\begin{example}
\begin{equation*}
\left(\!\!\begin{tabular}{c|c|c|c}
        $15$ & $2$ & $3$ & $4$ \\
        $45$ & $3$ & $2$ & $1$
    \end{tabular}\!\!\right) \circ 
\left(\!\!\begin{tabular}{c|c|c}
        $12$ & $34$ & $5$ \\
        $24$ & $35$ & $1$
    \end{tabular}\!\!\right) = 
\left(\!\!\begin{tabular}{c|c}
        $12$ & $345$ \\
        $13$ & $245$
    \end{tabular}\!\!\right)
\end{equation*}
\end{example}
\begin{remark}
In what follows, the result of the binary operation on two elements $x,y$ in an inverse semigroup is denoted simply by $xy$.
\end{remark}
With the defined (associative) composition, it is evident that the set of idempotents
\begin{equation}\label{eq:i(ubp)}
I(\mathcal{UBP}_n) \coloneqq \{x\in \mathcal{UBP}_n \, \vert \, xx = x\} = \{(\alpha, \alpha, id_{\alpha}) \in \mathcal{UBP}_n\}
\end{equation}
is in bijection with $\Pi_n$, where $id_{\alpha}:\alpha \rightarrow \alpha$ is the identity map. The composition then becomes commutative on $I(\mathcal{UBP}_n)$: $(\alpha, \alpha, id_{\alpha}) (\beta, \beta, id_{\beta}) = (\alpha \vee \beta, \alpha \vee \beta, id_{\alpha \vee \beta})$, where $\alpha \vee \beta$ is the join of $\alpha$ and $\beta$ in $\Pi_n$ (defined in the next subsection). For $x = (\alpha,\beta, f) \in \mathcal{UBP}_n$, consider $y = (\beta, \alpha, f^{-1})\in \mathcal{UBP}_n$, where $f^{-1}:\beta \rightarrow \alpha$ is inverse of $f:\alpha \rightarrow \beta$. It is then clear that $y=yxy$, $x=xyx$, and $y$ is denoted as $x^{-1}$. All the aforementioned properties imply that $\mathcal{UBP}_n$ is an inverse semigroup \cite{lawson-isemigroups, clifford-semigroups}.

\subsection*{Order structure}

\begin{definition}
A partially ordered set (poset) $P$ is a set together with a binary relation $\leq$ which is reflexive, antisymmetric and transitive. That is, for all $x,y,z \in P$, we have
\begin{enumerate}
\item $x\leq x$ (reflexivity)
\item $x\leq y, y\leq z \Rightarrow x\leq z$ (transitivity)
\item $x\leq y, y\leq x \Rightarrow x=y$ (antisymmetry)
\end{enumerate} 
\end{definition}

\begin{remark}
There may exist elements $x,y$ in a poset $P$ which are not comparable $x\nleq y$. Such $x,y$ are then said to be unrelated.
\end{remark}

A poset $P$ is said to contain $\mathbb{0}$ if $\exists \, \mathbb{0}\in P$ such that $\mathbb{0}\leq x, \forall x\in P$. A poset $P$ is said to contain $\mathbb{1}$ if $\exists \, \mathbb{1}\in P$ such that $x\leq \mathbb{1}, \forall x\in P$. For $x\leq y\in P$, define the interval $[x,y] \coloneqq \{z\in P \vert x\leq z\leq y\}$. $P$ is said to be locally finite if every interval in it is finite. $x\leq y$ with $x\neq y$ is denoted as $x<y$. Two posets $P$ and $Q$ are said to be isomorphic $P \cong Q$ if there exists an order preserving bijection $\phi : P \rightarrow Q$ such that for all $x,y \in P$, $x\leq_P y \Leftrightarrow \phi (x) \leq_Q \phi (y)$. For instance, given a non-empty set $X$, we have $\Pi_X \cong \Pi_{\vert X \vert}$.

\begin{definition}
A lattice $L$ is a poset such that for every pair $x,y\in L$, there exists a unique greatest lower bound (called meet) $x\wedge y$ and least upper bound (called join) $x\vee y$ in $L$. That is
\begin{align}
x\wedge y &= \mathrm{max}\{z\in L \, \vert \, z\leq x \, \mathrm{and} \, z\leq y\} \\ \nonumber
x\vee y &= \mathrm{min}\{z\in L \, \vert \, x\leq z \, \mathrm{and} \, x\leq z\}
\end{align}
\end{definition}

We now define partial order relations on $\Pi_n$ and $\mathcal{UBP}_n$. 
\begin{definition}[Refinement order on set partitions]
For $\alpha, \alpha' \in \Pi_n$, define $\alpha' \leq \alpha$ if each block of $\alpha'$ is contained in some block of $\alpha$. In other words, each block $\alpha_i \in \alpha$ can be written as a union of blocks $\alpha'_j \in \alpha'$: $\alpha_i = \displaystyle\cup_{j\in A_i} \alpha'_j$ where $A_i$ label the indices of the blocks of $\alpha'$ which participate in the union. 
\end{definition}
We read $\alpha' \leq \alpha$ as ``$\alpha'$ is a refinement of $\alpha$''. The finest and coarsest partitions in $\Pi_n$ are $\mathbb{0}_n=1\vert 2\vert \ldots\vert n$ and $\mathbb{1}_n=12\ldots n$ respectively. In $\Pi^{(2)}_{2n}$, the coarsest partition is still $\mathbb{1}_{2n}=12\ldots 2n$, but there are $P_{2n}(\lambda = 2^n)$ pairwise unrelated finest partitions consisting of blocks of size 2. 

\begin{definition}[Refinement order on uniform block permutations]\label{def:order-UBP}
For $(\alpha', \beta', f'), (\alpha, \beta, f)$ in $\mathcal{UBP}_n$, define $(\alpha', \beta', f') \leq (\alpha, \beta, f)$ if $\alpha' \leq \alpha$ and $\beta' \leq \beta$ as set partitions and $\alpha'_j \subseteq \alpha_i$ implies $f'(\alpha'_j) \subseteq f(\alpha_i)$ for all $\alpha'_j \in \alpha', \alpha_i \in \alpha$. 
\end{definition}

It is easy to see that ${\tiny\left(\!\! \begin{tabular}{c}
        1 2 3 $\cdots$ p \\
        1 2 3 $\cdots$ p
    \end{tabular} \!\!\right)}$ is the (unique) coarsest UBP, while all permutations of the form
    ${\tiny\left(\!\!\begin{tabular}{c|c|c|c}
        1 & 2 & $\cdots$ & p \\
        $\ast$ & $\ast$ & $\cdots$ & $\ast$
    \end{tabular}\!\!\right)}$ are the finest UBPs. Also notice that any two such permutations in $\mathcal{UBP}_n$ are not related and their meet does not exist in $\mathcal{UBP}_n$. This leads us to the following remark.

\begin{remark}
With the refinement order, $\Pi_n$ is a lattice, while $\mathcal{UBP}_n$ and $\Pi^{(2)}_{2n}$ are not.
\end{remark}

\begin{remark} \label{remark:naturalrefinement}
When considered as an inverse semigroup, the natural partial order on $\mathcal{UBP}_n$ \cite{lawson-isemigroups, clifford-semigroups} coincides with the dual of the refinement order on $\mathcal{UBP}_n$ introduced above (obtained by reversing all relations). 
\end{remark}

\subsection*{The M\"obius Function} 
For a locally finite poset $P$, identify its order with the graph $\leq \coloneqq \{(x,y)\in P \times P \, \vert \, x\leq y\}$. Define the incidence algebra of $P$ over $\mathbb{R}$ (denoted by $\mathsf{A}_{\mathbb{R}}(P)$) as the algebra of functions $g:\leq \rightarrow \mathbb{R}$ equipped with the convolution product. The M\"obius function $\mu_P$ on $P$ can then be defined as the inverse of the identity function $\zeta$ $[\zeta(x,y)=1 \, \forall x\leq y]$ in $\mathsf{A}_{\mathbb{R}}(P)$. The following serves as an inductive definition of $\mu_P$:
\begin{equation}
\mu_P(x,y) = \begin{cases} 
			1,& \text{if} \, x=y \\
 			-\sum_{x\leq z < y} \mu_P(x,z),& \text{if} \, x<y
 			\end{cases} 
\end{equation} 

We now state the M\"obius inversion formula \cite{stanley-combinatorics-1}, which will play a decisive role (along with the multiplicative families of combinatorial functions defined in the next subsection) in ensuring a succinct presentation of our graphical calculus in later sections. 

\begin{theorem}[M\"obius inversion formula]\label{thm:Mobius-inversion}
Let $P$ be a finite poset and $G$ be an abelian group. Then, for $f,g:P \rightarrow G$, the following holds
\begin{equation}
    g(x) = \displaystyle\sum_{y\geq x} f(y) \, \, \forall x\in P \quad \Leftrightarrow \quad f(x) = \displaystyle\sum_{y\geq x} \mu_P(x,y) g(y) \, \, \forall x\in P
\end{equation}

\end{theorem}

The M\"obius function on the partition lattice $\mu_{\Pi}$ is well known. For an interval $[\alpha, \beta] \subseteq \Pi_n$ where $\beta = \beta_1\vert \beta_2\vert \ldots\vert \beta_k$, we can write each block $\beta_i$ of $\beta$ as a union over blocks of $\alpha$: $\displaystyle\cup_{j\in B_i} \alpha_j$. Then, choosing a partition $\pi \in [\alpha, \beta]$ is equivalent to choosing a partition of the index set $B_i$ for $i=1,2,\ldots ,k$ which in turn is equivalent to choosing a partition in $\Pi_{b_i}$ for $i=1,2,\ldots ,k$ where $b_i = \vert B_i \vert \, \forall i$. We thus obtain
$[\alpha, \beta] \cong \Pi_{b_1} \times \Pi_{b_2} \times \ldots \times \Pi_{b_k}$, which implies 
\begin{equation}
\mu_{\Pi}(\alpha, \beta) = \prod_{i=1}^k \mu_{b_i}
\end{equation}
where the numbers $\mu_n$ are defined as $\mu_n \coloneqq \mu_{\Pi}(\mathbb{0}_n, \mathbb{1}_n)$. Application of Weisner's Theorem for M\"obius functions on finite lattices \cite{stanley-combinatorics-1} then yields $\mu_n = (-1)^{n-1}(n-1)!$ for $n\in \mathbb{N}$. Notice that if $\alpha, \beta \in \Pi_{2n}$ have even block sizes, then $[\alpha, \beta] \subseteq \Pi^{(2)}_{2n}$ and so the M\"obius function $\mu^{(2)}_{\Pi}$ on $\Pi^{(2)}_{2n}$ is just the restriction of the M\"obius function $\mu_{\Pi}$ on $\Pi_{2n}$ to $\Pi^{(2)}_{2n} \subseteq \Pi_{2n}$.

To compute the M\"obius function $\mu_\mathcal{U}$ on $\mathcal{UBP}_n$,  we take help from the following lemma \cite{steinberg-mobius}

\begin{lemma} \label{lemma:musemigroups}
The M\"obius function $\mu_S$ on an inverse semigroup $S$ endowed with the natural partial order and a finite set of idempotents $I(S)$ can be fully characterized in terms of the M\"obius function $\mu_{I(S)}$ on $I(S)$. That is, for all $x\leq y$ in $S$, we have 
\begin{equation}
\mu_S(x,y) = \mu_{I(S)}(xx^{-1},yy^{-1}) = \mu_{I(S)}(x^{-1}x,y^{-1}y)
\end{equation}  
\end{lemma}

\begin{proof}
The natural partial order on an inverse semigroup $S$ reads: $x\leq y \Leftrightarrow x=yx^{-1}x \Leftrightarrow x\in yI(S) \Leftrightarrow x\in I(S)y \Leftrightarrow x=xx^{-1}y$. For idempotents $e,f \in I(S)$, this order reduces to $e \leq f \Leftrightarrow e=ef=fe$. Clearly, the principle order ideal $y^\downarrow \coloneqq \{x\in S \, \vert \, x\leq y \}$ generated by $y\in S$ is equal to $yI(S)$, which is isomorphic as a poset to the set of idempotents $f\leq yy^{-1}$ through the order preserving bijection $\phi :  yI(S) \rightarrow \{f\in I(S) \, \vert \, f\leq yy^{-1} \}$ defined by $\phi(ye) = yey^{-1}$. This implies that for $x,y \in S$, the interval $[x,y] \cong \{ f\in I(S) \, \vert \, xx^{-1} \leq f \leq yy^{-1} \}$, which immediately gives the desired result.
\end{proof}

From Remark \ref{remark:naturalrefinement}, we know that the natural and refinement orders on $\mathcal{UBP}_n$ are just dual of each other. Moreover, since the set of idempotents in $\mathcal{UBP}_n$ is isomorphic as a poset to the partition lattice $\Pi_n$ (see Eq.~\eqref{eq:i(ubp)}), a straightforward application of Lemma \ref{lemma:musemigroups} yields
\begin{equation}
\mu_\mathcal{U} [(\alpha',\beta',f'),(\alpha,\beta,f)] = \mu_{\Pi}(\alpha',\alpha) = \mu_{\Pi}(\beta',\beta)
\end{equation}

\begin{remark} \label{remark:2}
The order $n$ of the underlying posets $\Pi_n$ and $\mathcal{UBP}_n$ is implicit while writing the relevant M\"obius functions $\mu_{\Pi}$ and $\mu_\mathcal{U}$. 
\end{remark}

\subsection*{Multiplicative families of Combinatorial Functions}
We now introduce a family of combinatorial functions on $\mathcal{UBP}_n$ and $\Pi^{(2)}_{2n}$ for $n\in \mathbb{N}$.
\begin{definition}\label{def:Cfu} 
Define $\Cf_\mathcal{U}:\mathcal{UBP}_n \rightarrow \mathbb{R}$ by 
\begin{align}
\Cf_\mathcal{U}(\alpha, \beta, f) &= \displaystyle\sum_{(\alpha', \beta', f') \leq (\alpha, \beta, f)} \mu_\mathcal{U} [(\alpha',\beta',f'),(\alpha,\beta,f)] \nonumber \\
&= \displaystyle\sum_{(\alpha', \beta', f')\leq (\alpha,\beta,f)} \mu_\Pi (\alpha', \alpha)
\end{align}
where the sum runs over all refinements $(\alpha', \beta', f')$ of $(\alpha, \beta, f)$ in $\mathcal{UBP}_n$.
\end{definition}

\begin{definition}\label{def:Cfpi} 
Define $\Cf_\Pi:\Pi^{(2)}_{2n} \rightarrow \mathbb{R}$ by 
\begin{equation}
\Cf_\Pi(\alpha) = \displaystyle\sum_{\alpha'\leq \alpha} \mu_\Pi (\alpha', \alpha) \nonumber 
\end{equation}
where the sum runs over all even refinements $\alpha'$ of $\alpha$ in $\Pi^{(2)}_{2n}$.
\end{definition}

\begin{lemma}[Multiplicativity of Combinatorial functions on $\mathcal{UBP}_n$] \label{lemma:multi-comb-func-ubp}
Let $(\alpha, \beta, f) \in \mathcal{UBP}_n$. Index the blocks of $\alpha$ such that $\alpha = \alpha_1\vert \alpha_2\vert \ldots\vert \alpha_k$. Let $a_i = \vert \alpha_i \vert \, \forall i$. Then,
\begin{align}
    \Cf_\mathcal{U}(\alpha, \beta, f) &= \displaystyle\prod_{i=1}^k \left\{ \displaystyle\sum_{\lambda_i \vdash a_i} U_{a_i}(\lambda_i) \mu_{\Pi} (\pi^{\lambda_i}_{a_i}, \mathbb{1}_{a_i}) \right\} \nonumber \\
    &= \displaystyle\prod_{i=1}^k \Cf_\mathcal{U}(\mathbb{1}_{a_i}, \mathbb{1}_{a_i}, id_{a_i})
\end{align}
where $\pi^{\lambda_i}_{a_i} \in \Pi_{a_i}$ is an arbitrary partition of type $\lambda_i \vdash a_i$.
\end{lemma}

\begin{proof}
Firstly, notice that since $(\mathbb{1}_n,\mathbb{1}_n,id_n)$ is coarsest in $\mathcal{UBP}_n$, every $(\alpha' ,\beta', f')\in \mathcal{UBP}_n$ is a refinement. Corresponding to each integer partition $\lambda \vdash n$ then, there are $U_n(\lambda)$ refinements (see Eq.~\eqref{eq:num-ubp}), each with a M\"obius factor $\mu_{\Pi}(\pi^{\lambda}_n, \mathbb{1}_n) = (-1)^{m_\lambda -1}(m_\lambda -1)!$, where $m_\lambda$ is the total number of blocks in $\pi^{\lambda}_n$. Thus, it is clear that 
\begin{equation}
    \Cf_\mathcal{U}(\mathbb{1}_n,\mathbb{1}_n,id_n) = \displaystyle\sum_{(\alpha', \beta', f')\leq (\mathbb{1}_n, \mathbb{1}_n, id_n)} \mu_\Pi (\alpha',\mathbb{1}_n) = \displaystyle\sum_{\lambda \vdash n} U_n(\lambda) \mu_\Pi (\pi^{\lambda}_n, \mathbb{1}_n)
\end{equation}
Now, for a general $(\alpha, \beta, f)\in \mathcal{UBP}_n$, notice that $(\alpha', \beta' ,f')\leq (\alpha, \beta, f)$ implies that each block of $\alpha$ can be written as a union of blocks of $\alpha'$: $\alpha_i = \cup_{j\in A_i} \alpha'_j$ for $i=1,2,\ldots ,k$. This allows us to label the refinements $(\alpha', \beta', f')$ by a sequence of $k$ partition types corresponding to each block $\alpha_i$: $\{\lambda_i \vdash a_i\}_{i=1}^k$. Moreover, for each type $\{\lambda_i \vdash a_i\}_{i=1}^k$, there are a total of $\prod_{i=1}^k U_{a_i}(\lambda_i)$ refinements, each with a M\"obius factor $\prod_{i=1}^k \mu_\Pi (\pi^{\lambda_i}_{a_i}, \mathbb{1}_{a_i})$. Summing these contributions over all partition types $\{\lambda_i \vdash a_i\}_{i=1}^k$ yields the desired result:
\begin{align}
    \Cf_\mathcal{U}(\alpha, \beta, f) &= \mathlarger{\mathlarger{\sum}}_{ \left\{  \substack{ \lambda_i \vdash a_i \\ . \\ . \\ \lambda_k \vdash a_k }      \right\}} \left\{ \displaystyle\prod_{i=1}^k U_{a_i}(\lambda_i) \mu_\Pi (\pi^{\lambda_i}_{a_i}, \mathbb{1}_{a_i}) \right\} \nonumber \\
    &= \displaystyle\prod_{i=1}^k \left\{ \displaystyle\sum_{\lambda_i \vdash a_i} U_{a_i}(\lambda_i) \mu_{\Pi} (\pi^{\lambda_i}_{a_i}, \mathbb{1}_{a_i}) \right\} \nonumber \\
    &= \displaystyle\prod_{i=1}^k \Cf_\mathcal{U}(\mathbb{1}_{a_i}, \mathbb{1}_{a_i}, id_{a_i}) \nonumber
\end{align}
\end{proof}
We take a moment to explicitly state the formula for the numbers $\Cf_\mathcal{U}(\mathbb{1}_n,\mathbb{1}_n,id_n) \coloneqq \Cf_\mathcal{U}(n) $. 
\begin{equation}
    \Cf_\mathcal{U}(\mathbb{1}_n,\mathbb{1}_n,id_n) = \sum_{\lambda \vdash n} \left( \frac{n!}{(1!)^{m_1}(2!)^{m_2}\ldots (n!)^{m_n}} \right)^2 \frac{(-1)^{m_\lambda -1} (m_\lambda -1)!}{m_1!m_2!\ldots m_n!} 
\end{equation}
Here, $m_\lambda = \sum_{i=1}^n m_i$ is the total number of blocks in an arbitrary partition of type $\lambda = 1^{m_1}2^{m_2}\ldots n^{m_n}$ in $\Pi_n$. The sum runs over all integer partitions $\lambda \vdash n$. We state the initial values of the sequence $\{ \Cf_\mathcal{U}(\mathbb{1}_n,\mathbb{1}_n,id_n) \}_{n\in \mathbb{N}}$ below (\cite{oeis} \href{https://oeis.org/A101981}{A101981}).
\begin{equation}
    1,-1,4,-33,456,-9460,274800,\ldots
\end{equation}

The combinatorial functions on $\Pi^{(2)}_{2n}$ also possess a similar multiplicative property. The proof of Lemma \ref{lemma:multi-comb-func-ubp} can be easily altered to exhibit this.

\begin{lemma}[Multiplicativity of Combinatorial functions on $\Pi^{(2)}_{2n}$]
Let $\alpha \in \Pi^{(2)}_{2n}$. Index the blocks of $\alpha$ such that $\alpha = \alpha_1\vert \alpha_2\vert \ldots\vert \alpha_k$. Let $a_i = \vert \alpha_i \vert \, \forall i$. Then,
\begin{align}
    \Cf_\Pi(\alpha) &= \displaystyle\prod_{i=1}^k \left\{ \displaystyle\sum_{\lambda_i \Vdash a_i} P_{a_i}(\lambda_i) \mu_{\Pi} (\pi^{\lambda_i}_{a_i}, \mathbb{1}_{a_i}) \right\} \nonumber \\
    &= \displaystyle\prod_{i=1}^k \Cf_\Pi(\mathbb{1}_{a_i})
\end{align}
where $\pi^{\lambda_i}_{a_i} \in \Pi^{(2)}_{a_i}$ is an arbitrary even partition of type $\lambda_i \Vdash a_i$.
\end{lemma}

The numbers $\Cf_\Pi(\mathbb{1}_{2n}) \coloneqq \Cf_\Pi(n)$ are given by the formula 
\begin{equation}
    \Cf_\Pi(\mathbb{1}_{2n}) = \sum_{\lambda \Vdash 2n } \frac{(2n)! (-1)^{a_\lambda -1}(a_\lambda -1)!}{a_1!a_2!\ldots a_n!(2!)^{a_1}(4!)^{a_2}\ldots ((2n)!)^{a_n}}   
\end{equation}
Here, $a_\lambda = \sum_{i=1}^n a_i$ is the total number of blocks in an arbitrary even partition of type $\lambda = 2^{a_1}4^{a_2}\ldots (2n)^{a_n}$ in $\Pi^{(2)}_{2n}$. The sum runs over all even integer partitions $\lambda \Vdash 2n$. We state the initial values of the sequence $\{ \Cf_\Pi(\mathbb{1}_{2n}) \}_{n\in \mathbb{N}}$ below (\cite{oeis} \href{https://oeis.org/A000182}{A000182}).
\begin{equation}
    1,-2,16,-272,7936,-353792,22368256,\ldots
\end{equation}

\begin{remark}
As was the case with the M\"obius functions, the underlying order $n \in \mathbb{N}$ of $\Pi^{(2)}_{2n}$ and $\mathcal{UBP}_n$ is tacit while writing the relevant combinatorial functions $\Cf_\Pi$ and $\Cf_\mathcal{U}$. 
\end{remark}

\begin{remark}
The combinatorial considerations from this section section fit in nicely in Stanley's theory of \emph{exponential structures} \cite{stanley1978exponential} and \cite[Section 5.5]{ stanley-combinatorics-2}. Indeed, it can easily be seen from the previous discussion that the sequence of posets $(\mathcal{UBP}_n)_n$ (resp.~$(\Pi^{(2)}_{2n})_n$) is an exponential structure with denominator sequence $M(n)= n!$ (resp.~$M(n) =  (2n)!!=(2n-1)(2n-3)\cdots5\cdot3\cdot1$). The coefficients $\Cf_\mathcal{U}$ and $\Cf_\Pi$ satisfy the following expressions (see \cite[Eq.~(8)]{stanley1978exponential}, with $\Cf(n) = -\mu_n$):
\begin{align*}
    \sum_{n=1}^\infty \frac{\Cf_{\mathcal U}(n) x^n}{n!^2} &= \log \sum_{n=0}^\infty \frac{x^n}{n!^2} = \log I_0(2\sqrt{x})\\
    \sum_{n=1}^\infty \frac{\Cf_{\Pi}(n) (2x)^n}{(2n)!} &= \log \sum_{n=0}^\infty \frac{(2x)^n}{(2n)!} = \log \cosh(\sqrt{2x}),
\end{align*}
where $I_0$ is the \emph{modified Bessel function of the first kind}.
\end{remark}

\section{Diagrammatic notation for tensors}\label{sec:diagrams-tensors}

We describe in this section a graphical way of representing tensors which will be at the heart of the integration formulas in Sections \ref{sec:graphical-integration-C} and \ref{sec:graphical-integration-R}. The idea of representing tensors (resp.~tensor contractions) with boxes (resp.~wires) goes back to Penrose \cite{penrose1971applications}. Modern presentations, related to quantum information theory and tensor networks can be found in \cite{wood2015tensor,bridgeman2017hand,coecke2017picturing}; we refer the reader interested in the full power of the formalism to these excellent references. 

In what follows, uppercase letters are used to denote matrices $X,Y \in \mathcal{M}_d(\mathbb{C})$ and lowercase letters are used to denote vectors $v,w \in \mathbb{C}^d$ or scalars $x,y \in \mathbb{C}$, depending on the context. Given $v \in \mathbb{C}^d$ (resp. $X \in \mathcal{M}_d(\mathbb C)$), conjugate transpose is denoted by $v^*$ (resp. $X^*$) while entrywise complex conjugate is denoted by $\overbar{v}$ (resp. $\overbar{X}$). Hadamard (or entrywise) product of two matrices $X,Y \in \mathcal{M}_d(\mathbb C)$ is denoted by $X\odot Y$. The standard basis in $\mathbb{C}^d$ (resp. $\mathbb{C}^d \otimes \mathbb{C}^d$) is denoted by $\{e_i\}_{i=1}^d$ (resp. $\{ e_i \otimes e_j \}_{i,j=1}^d$).

\bigskip

The building blocks of the diagrammatic notation for tensors are \emph{boxes}, which encode tensors. Typically, we depict a $k$-tensor by a box with $k$ half-edges sticking out of it. We usually represent edges corresponding to primal spaces (vectors, or kets) pointing to the left, while edges corresponding to dual spaces (linear forms, or bras) point to the right. We depict in Figure \ref{fig:tensors} some examples of tensors. 

\begin{figure}[hbt!]
    \centering
\includegraphics[align=c]{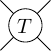}\qquad\qquad\qquad
\includegraphics[align=c]{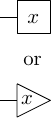}\qquad\qquad\qquad
\includegraphics[align=c]{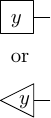}\qquad\qquad\qquad
\includegraphics[align=c]{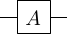}
    \caption{In the diagrammatic notation, tensors are represented by boxes. From left to right: a 4-tensor $T \in (\mathbb C^d)^{\otimes 4}$; a vector $x \in \mathbb C^d$; a dual vector $y \in (\mathbb C^d)^*$; a matrix $A \in \mathcal M_d(\mathbb C)$.}
    \label{fig:tensors}
\end{figure}

One can combine two or more diagrams by drawing them one next to the other. This operation corresponds to taking the \emph{tensor product} of the respective diagrams, see Figure \ref{fig:tensor-product}

\begin{figure}[hbt!]
    \centering
\includegraphics[align=c]{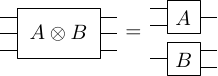} \qquad\qquad\qquad
\includegraphics[align=c]{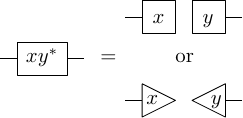}
    \caption{Right: the tensor product of two diagrams corresponds to diagram juxtaposition; here, we depict $A \otimes B \in \mathcal M_{d^3}(\mathbb C)$. Left: the rank one matrix $xy^* \in \mathcal M_d(\mathbb C)$, for $x,y \in \mathbb{C}^d$}
    \label{fig:tensor-product}
\end{figure}

The next important feature of the diagrammatic notation for tensors is the fact that tensor contraction corresponds to adding a \emph{wire} to the corresponding half-edges attached to the tensors. Using coordinates, adding a wire requires that the two indices corresponding to the two ends of the wire must be identical, and that there is a sum over the common index, see Figure \ref{fig:wires}. Mathematically, tensor contractions correspond to the evaluation map
\begin{align*}
    V^* \otimes V &\to \mathbb C\\
    (\phi, x) &\mapsto \langle \phi,x \rangle = \phi(x).
\end{align*}

\begin{figure}[hbt!]
    \centering
    \includegraphics[align=c]{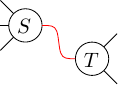} \qquad\qquad\qquad
    \includegraphics[align=c]{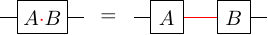}\qquad\qquad\qquad
    \includegraphics[align=c]{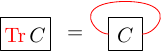}
    \caption{Left: a \textcolor{red}{tensor contraction} between two tensors $S,T$. Center: the matrix product is a tensor contraction. Right: the trace of a matrix $C \in \mathcal M_d(\mathbb C)$.}
    \label{fig:wires}
\end{figure}

 We discuss now a special type of tensors, made only out of wires, see Figure \ref{fig:wire-tensors}. One reads these tensors by interpreting each wire as a delta function, requiring that the corresponding coordinates match. 

\begin{figure}[hbt!]
    \centering
    \includegraphics[align=c]{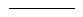} \qquad\qquad\qquad
    \includegraphics[align=c]{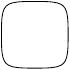}\qquad\qquad\qquad
    \includegraphics[align=c]{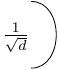}\qquad\qquad\qquad
    \includegraphics[align=c]{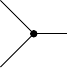}
    \caption{From left to right: the identity matrix; loops correspond to the scalar $d = \dim V$; a maximally entangled state $ \Omega = \frac{1}{\sqrt d} \sum_{i=1}^d {e_i \otimes e_i}$; an un-normalized GHZ state $\mathrm{GHZ} = \sum_{i=1}^d e_i \otimes e_i \otimes e_i$.}
    \label{fig:wire-tensors}
\end{figure}

Adding extra wires to diagrams representing matrices or vectors allows us to represent the conditional expectation on the diagonal matrix algebra and related quantities, see Figure \ref{fig:diagonal}. 

\begin{figure}[hbt!]
    \centering
    \includegraphics[align=c]{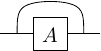} \qquad\qquad\qquad
    \includegraphics[align=c]{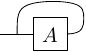} \qquad\qquad\qquad
    \includegraphics[align=c]{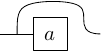}
    \caption{Left: the conditional expectation onto the algebra of diagonal matrices applied to $A$; in other words, the matrix obtained by setting the off-diagonal entries of $A$ to zero. Center: the diagonal vector of a matrix $A$. Right: the diagonal matrix with the vector $a$ on its diagonal.}
    \label{fig:diagonal}
\end{figure}

The transposition and the partial transposition of matrices will play an important role in this paper. The transposition operation can be represented graphically by permuting the input and output dangling edges of a matrix. The partial transposition (resp.~the partial trace) operations have equally pleasant graphical representations, see Figure \ref{fig:transpositions}. 

\begin{figure}[hbt!]
    \centering
    \includegraphics[align=c]{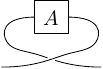} \qquad\qquad\qquad
    \includegraphics[align=c]{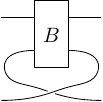} \qquad\qquad\qquad
    \includegraphics[align=c]{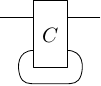}
    \caption{From left to right: the transpose $\operatorname{transp}(A) := A^\top$; the partial transpose $B^\Gamma := [\operatorname{id} \otimes \operatorname{transp}](B)$; the partial trace $[\operatorname{id} \otimes \operatorname{Tr}](C)$.}
    \label{fig:transpositions}
\end{figure}

We end this section by mentioning that we shall address a specific assignment of coordinates to a tensor by writing index values on top of dangling edges, as in Figure \ref{fig:coordinates}. 

\begin{figure}[hbt!]
    \centering
    \includegraphics[align=c]{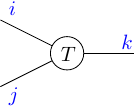} \qquad\qquad\qquad
    \includegraphics[align=c]{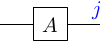}
    \caption{Integers written on top of the dangling edges of a tensor correspond to assigning specific coordinates. On the left panel we have the coordinate $T_{ijk} \in \mathbb C$ of a 3-tensor $T$. On the right panel, we depict the diagram of the $j$-th column vector of a matrix $A$. }
    \label{fig:coordinates}
\end{figure}

\section{Graphical integration formula --- complex phases}\label{sec:graphical-integration-C}

This section contains the main theoretical result of our work, a \emph{graphical integration formula}, which enables one to compute expectation values of diagrams with respect to random variables comprised of independent and identically distributed complex phases. 

Let us start by describing the probability distribution of the random variables we shall consider. In this paper, $\mathbb T$ will denote the unit circle $\mathbb T = \{z \in \mathbb C \, : \, |z| = 1\}$. 

\begin{definition}\label{def:rv-u}
A \emph{random phase} is a random variable $z \in \mathbb T$ having uniform distribution on the unit circle $z = \exp(\mathrm{i} \theta)$, where $\theta$ is uniformly distributed on $[0,2\pi]$.
A \emph{random phase vector} is a random variable $u \in \mathbb T^d$ having independent and identically distributed coordinates
\begin{equation}\label{eq:def-rv-u}
    u=(u_1, \ldots, u_d) \text{ with } u_k = \exp(\mathrm{i}\theta_k) \text{ where $\theta_1, \ldots, \theta_d$ are i.i.d. uniform in $[0,2\pi]$.}
\end{equation}
\end{definition}

\begin{figure}[hbt!]
    \centering
    \includegraphics{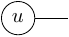} \qquad \qquad \includegraphics{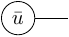}
    \caption{Diagrams corresponding to random phase vectors and their conjugates. The distribution of the random variable $u$ is given in \eqref{eq:def-rv-u}.}
    \label{fig:boxes-u-ubar}
\end{figure}

\vspace{0.5cm}

Given a diagram $\mathcal D$ containing boxes corresponding to random variables $u$ and $\overbar u$, our goal is to express the expectation value of $\mathcal D$ (with respect to the distribution of $u$, denoted by $\mathbb E_u \mathcal{D}$) as a weighted sum of diagrams obtained by gluing particular sets of edges of $\mathcal D$. The first step in this direction is to pull the $u$ and $\overbar u$ boxes outside, so as to \emph{separate} them from the rest of the diagram, see Figure \ref{fig:separate-uubar-D}. Note that this procedure does not modify the tensor represented by $\mathcal D$, it is just an aesthetical rearrangement of the boxes composing $\mathcal D$. In this way, $\mathcal D$ can be seen as the contraction of a diagram $\mathcal D^\circ$ with a bunch of $u$ (usually represented on the left of $\mathcal D^\circ$) and $\overbar u$ (usually represented on the right) boxes, see Figure~\ref{fig:separate-uubar-D}.

\begin{figure}[H] 
    \centering
    \includegraphics{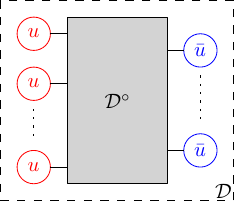}
    \caption{Given a diagram $\mathcal D$, we pull the \textcolor{red}{$u$} boxes out to the left, the \textcolor{blue}{$\overbar u$} boxes to the right, the remaining boxes forming a diagram $\mathcal D^\circ$.}
    \label{fig:separate-uubar-D}
\end{figure}

Before describing the graphical integration procedure, let us show that the only relevant cases are the ones where the number of $u$-boxes is equal to the number of $\overbar u$-boxes. 

\begin{lemma} \label{lemma:n neq m = 0}
Let $\mathcal{D}$ be a diagram containing $n$ $u$-boxes and $m$ $\overbar{u}$-boxes, corresponding to random variables $u$ as in \eqref{eq:def-rv-u}. If $n \neq m$, then $\mathbb{E}_u \mathcal{D} = 0$.
\end{lemma}

\begin{proof}
Let $\omega \in \mathbb T$ be a fixed arbitrary phase. The distribution of the random phase vector $u$ from Definition \ref{def:rv-u} is invariant under a (global) rotation by $\omega$: $u \stackrel{dist}{=} \omega u$. Hence, we have $\mathbb E_u \mathcal D = \omega^{n-m}\mathbb E_u \mathcal D$, which, by choosing an appropriate value for $\omega$, implies $\mathbb E_u \mathcal D=0$ as claimed. 

\end{proof}

From now on, we shall only consider diagrams having an equal number of $u$ and $\overbar u$ boxes. As we will see shortly, our main result allows one to write the expectation $\mathbb{E}_u \mathcal D$ as a weighted sum of diagrams $\mathcal D_{(\alpha, \beta, f)}$, obtained by gluing the legs attached to the $u$ and the $\overbar u$ boxes in a manner prescribed by the uniform block permutation $(\alpha, \beta, f) \in \mathcal{UBP}_n$, see Figure \ref{fig:example-E-u}. The following definition describes in detail how the diagrams $\mathcal D_{(\alpha, \beta, f)}$ are constructed.

\begin{definition}\label{def:D-UBP}
Given a diagram $\mathcal D$ containing $n$ $u$-boxes and $n$ $\overbar u$-boxes along with a uniform block permutation $(\alpha,\beta,f) \in \mathcal{UBP}_n$, the following procedure is followed to construct the diagram $\mathcal D_{(\alpha, \beta, f)}$: 
\begin{enumerate}
    \item remove the $u$ and the $\overbar u$ boxes
    \item connect the $n$ dangling wires previously attached to the $u$-boxes along the partition $\alpha$
    \item connect the $n$ dangling wires previously attached to the $\overbar u$-boxes along the partition $\beta$
    \item connect the blocks of the partitions $\alpha$ and $\beta$ using the permutation $f$
\end{enumerate}
\end{definition}

\begin{figure}[hbt!]
    \centering
    \includegraphics[align=c]{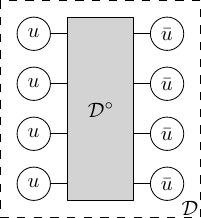} \qquad\qquad \includegraphics[align=c]{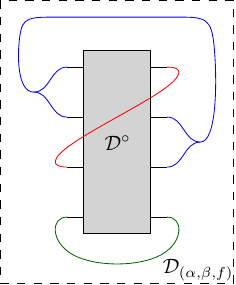} \qquad\qquad \includegraphics[align=c]{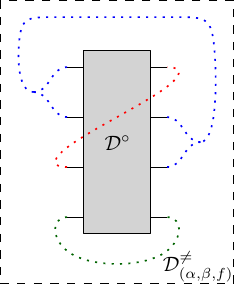}
    \caption{Left: a diagram $\mathcal D$ containing $n=4$ $u$-boxes and $m=4$ $\overbar u$-boxes. Center: the associated diagram $\mathcal D_{(\alpha, \beta, f)}$ corresponding to a UBP $(\alpha, \beta, f) \in \mathcal{UBP}_4$. The partition $\alpha = (\textcolor{blue}{12}|\textcolor{red}{3}|\textcolor{darkgreen}{4})$ is used to pair the wires corresponding to the $u$-boxes (on the left), and $\beta = (\textcolor{red}{1}|\textcolor{blue}{23}|\textcolor{darkgreen}{4})$ is used to pair the wires corresponding to the $\overbar u$ boxes (on the right). The bijective function $f$, matching the blocks of $\alpha$ and $\beta$ ($\textcolor{blue}{12} \leftrightarrow \textcolor{blue}{23}$, $\textcolor{red}{3} \leftrightarrow \textcolor{red}{1}$, $\textcolor{darkgreen}{4} \leftrightarrow \textcolor{darkgreen}{4}$) connects the blocks on the left with the blocks on the right. The wires are indexed from top to bottom. Right: the injective diagram $\mathcal D_{(\alpha,\beta,f)}^{\neq}$ corresponding to the same UBP. The indices corresponding to the dotted lines must be different.}
    \label{fig:example-E-u}
\end{figure}

Algebraically, after removing the $u$ and $\overbar u$-boxes from $\mathcal{D}$, we are left with the internal diagram $\mathcal D^\circ$ having $2n$ dangling edges, corresponding to a $2n$-tensor. Contracting the tensor $\mathcal D^\circ$ with the basis vectors corresponding to multi-indices $i,j:[n] \to [d]$ (see Figure~\ref{fig:coordinates}) yields the $\mathcal D^\circ_{i,j}$ diagram, allowing us to write

\begin{equation} \label{eq:D-vs-Dcirc-u}
\mathcal{D} = \sum_{i,j:[n] \to [d]} u_{i(1)}u_{i(2)}\cdots u_{i(n)} \overbar u_{j(1)}\overbar u_{j(2)} \cdots \overbar u_{j(n)} \mathcal{D}_{ij}^\circ.
\end{equation}

The construction procedure listed in Definition~\ref{def:D-UBP} then describes the following equation
\begin{equation}
\mathcal D_{(\alpha, \beta, f)} = \sum_{(i,j) \in \ker(\alpha,\beta,f)} \mathcal D^\circ_{i,j}
\end{equation}

where, instead of contracting $\mathcal{D}^{\circ}$ with arbitrary pairs of multi-indices $(i,j)$, we restrict ourselves to the set $\ker(\alpha,\beta,f)$, which consists of only those pairs of multi-indices $(i,j)$ which respect the given UBP $(\alpha,\beta, f)$ in the following sense:
\begin{itemize}
    \item if $x,y \in [n]$ belong to the same block of $\alpha$, then $i(x) = i(y)$, i.e., $i \in \ker (\alpha)$.
    \item if $x,y \in [n]$ belong to the same block of $\beta$, then $j(x) = j(y)$, i.e., $j \in \ker (\beta)$.
    \item if $x$ belongs to a block $a$ of $\alpha$, $y$ belongs to a block $b$ of $\beta$, and $f(a) = b$, then $i(x) = j(y)$.
\end{itemize}
For example, the diagram in the center panel of Figure \ref{fig:example-E-u} corresponds to the sum
\begin{equation}\label{eq:ex-D-sum}
\mathcal D_{\tiny \left(\!\!\begin{tabular}{c|c|c}
        \textcolor{blue}{1 2} & \textcolor{red}{3} & \textcolor{darkgreen}{4} \\
        \textcolor{blue}{2 3} & \textcolor{red}{1} & \textcolor{darkgreen}{4}
    \end{tabular}\!\!\right)}= \sum_{\textcolor{blue}{i},\textcolor{red}{j},\textcolor{darkgreen}{k}\in [d]} \mathcal D^\circ_{(\textcolor{blue}{i},\textcolor{blue}{i},\textcolor{red}{j},\textcolor{darkgreen}{k}),(\textcolor{red}{j},\textcolor{blue}{i},\textcolor{blue}{i},\textcolor{darkgreen}{k})}
\end{equation}

We now introduce a different type of diagrams $\mathcal{D}_{(\alpha, \beta, f)}^{\neq}$, which are obtained from the existing $\mathcal{D}_{(\alpha, \beta, f)}$ diagrams by imposition of an additional injectivity constraint.

\begin{definition}\label{def:D-UBP-injective}
Given a diagram $\mathcal D$ containing $n$ $u$-boxes and $n$ $\overbar{u}$-boxes along with a uniform block permutation $(\alpha,\beta,f) \in \mathcal{UBP}_n$, the \emph{injective diagram} $\mathcal D_{(\alpha, \beta, f)}^{\neq}$ corresponds to the tensor
$$\mathcal D_{(\alpha, \beta, f)}^{\neq} = \sum_{(i,j) \in \ker^{\neq}(\alpha,\beta,f)} \mathcal D^\circ_{i,j},$$
where $\ker^{\neq}(\alpha,\beta,f)$ is the set of \emph{injective} pairs of multi-indices $(i,j)$ respecting the UBP $(\alpha,\beta, f)$ in the following sense:
\begin{itemize}
    \item if $x,y \in [n]$ belong to different blocks of $\alpha$, then $i(x) \neq i(y)$
    \item if $x,y \in [n]$ belong to different blocks of $\beta$, then $j(x) \neq j(y)$.
    \item if $x,y \in [n]$ belong to the same block of $\alpha$, then $i(x) = i(y)$, i.e., $i \in \ker (\alpha)$.
    \item if $x,y \in [n]$ belong to the same block of $\beta$, then $j(x) = j(y)$, i.e., $j \in \ker (\beta)$.
    \item if $x$ belongs to a block $a$ of $\alpha$, $y$ belongs to a block $b$ of $\beta$, and $f(a) = b$, then $i(x) = j(y)$.
\end{itemize}
Graphically, the blocks of $\alpha$ and $\beta$ are connected (according to $f$) using dotted lines. Algebraically, the indices corresponding to different dotted lines should be different. 
\end{definition}

\begin{remark}
If the number of blocks in $(\alpha, \beta, f) \in \mathcal{UBP}_n$ exceeds the dimension $d$ of the underlying space $\mathbb{T}^d$ (in which the random vectors $u$ reside), the conditions stated in Definition~\ref{def:D-UBP-injective} become unattainable for any multi-index pair $i,j:[n]\rightarrow [d]$, i.e., $\ker^{\neq}(\alpha, \beta, f) $ is empty. Hence, $\mathcal{D}_{(\alpha, \beta, f)}^{\neq} = 0.$
\end{remark}

The injective diagram in the right panel of Figure \ref{fig:example-E-u} corresponds to the sum (compare to Eq.~\eqref{eq:ex-D-sum})
\begin{equation}
    \mathcal D^{\neq}_{\tiny \left(\!\!\begin{tabular}{c|c|c}
        \textcolor{blue}{1 2} & \textcolor{red}{3} & \textcolor{darkgreen}{4} \\
        \textcolor{blue}{2 3} & \textcolor{red}{1} & \textcolor{darkgreen}{4}
    \end{tabular}\!\!\right)}= \sum_{\substack{\textcolor{blue}{i},\textcolor{red}{j},\textcolor{darkgreen}{k}\in [d] \\ \textcolor{blue}{i} \neq \textcolor{red}{j} \neq \textcolor{darkgreen}{k}}} \mathcal D^\circ_{(\textcolor{blue}{i},\textcolor{blue}{i},\textcolor{red}{j},\textcolor{darkgreen}{k}),(\textcolor{red}{j},\textcolor{blue}{i},\textcolor{blue}{i},\textcolor{darkgreen}{k})}
\end{equation}

We are now in a position to state an important lemma, which relates the expectation value of a diagram to the sum over all injective diagrams.

\begin{lemma}\label{lem:Eu-sum-injectve-diagrams}
Let $\mathcal D$ be a diagram containing $n$ $u$-boxes and $n$ $\overbar u$-boxes. Then,
\begin{equation}\label{eq:Eu-sum-injective-diagrams}
\mathbb{E}_u \mathcal{D} = \sum_{(\alpha, \beta, f)\in \mathcal{UBP}_n} \mathcal{D}^{\neq}_{(\alpha, \beta, f)} 
\end{equation}
\end{lemma}

\begin{proof}
Using Eq.~\eqref{eq:D-vs-Dcirc-u} and the linearity of the expectation, it is clear that 
\begin{equation}\label{eq:Eu-Dcirc}
\mathbb{E}_u \mathcal{D} = \sum_{i,j:[n] \to [d]} \mathbb{E}_u[u_{i(1)}\cdots u_{i(n)} \overbar u_{j(1)} \cdots \overbar u_{j(n)}] \mathcal{D}_{ij}^\circ
\end{equation}
The expectation value $\mathbb{E}_u[u_{i(1)}\cdots u_{i(n)} \overbar u_{j(1)} \cdots \overbar u_{j(n)}]$ in the equation above is non-zero (therefore equal to 1) if and only if all the $i(\cdot)$ indices are paired to the $j(\cdot)$ indices---we shall call such $i$ and $j$ \emph{matching}. Such pairings are implemented by uniform block permutations: each $(\alpha, \beta, f) \in \mathcal{UBP}_n$ corresponds to a possible pairing, and each pairing configuration can be obtained in this way. Hence, the non-zero terms in the sum \eqref{eq:Eu-Dcirc} can be grouped according to their pairing configurations $(\alpha, \beta, f) \in \mathcal{UBP}_n$. In order to avoid double counting of pairings, one needs to assign to each UBP the matching indices $(i,j)$ corresponding to the \emph{injective} diagrams $\mathcal D^{\neq}_{(\alpha, \beta, f)}$. More precisely, the sets $\ker^{\neq}(\alpha, \beta, f)$ form a partition of the set of matching $(i,j)$, and we have
\begin{align*}
    \mathbb{E}_u \mathcal{D} &= \sum_{\substack{i,j:[n] \to [d] \\ (i,j) \text{ matching}}}  \mathcal{D}_{ij}^\circ \\
    &= \sum_{(\alpha, \beta, f)\in \mathcal{UBP}_n} \qquad \sum_{(i,j) \in \ker^{\neq}(\alpha, \beta, f)} \mathcal{D}_{ij}^\circ \\
    &=  \sum_{(\alpha, \beta, f)\in \mathcal{UBP}_n} \mathcal{D}^{\neq}_{(\alpha, \beta, f)}.
\end{align*}

\end{proof}

In order to express the desired expectation in terms of the (more intuitive) non-injective diagrams $\mathcal{D}_{(\alpha, \beta, f)}$, our next result relates the two types of diagrams. We do this by exploiting the order structure on the poset $\mathcal{UBP}_n$, see Definition \ref{def:order-UBP}.

\begin{lemma} \label{lem:Dabf-vs-Dabf-injective}
For any uniform block permutation $(\alpha, \beta, f) \in \mathcal{UBP}_n$, the following relations hold:
\begin{align}
    \mathcal{D}_{(\alpha, \beta, f)} &= \sum_{(\alpha', \beta', f') \geq (\alpha, \beta, f)} \mathcal{D}^{\neq}_{(\alpha', \beta', f')}\\
    \mathcal{D}^{\neq}_{(\alpha, \beta, f)} &=\sum_{(\alpha', \beta', f') \geq (\alpha, \beta, f)} \mathcal{D}_{(\alpha', \beta', f')} \mu_{\mathcal{U}} [(\alpha,\beta,f),(\alpha',\beta',f')].
\end{align}
\end{lemma}

\begin{proof}
The first equation is a restatement of the fact that the sets $\ker^{\neq}(\alpha', \beta',f)$, with $(\alpha', \beta', f') \geq (\alpha, \beta, f)$ form a partition of the set $\ker(\alpha,\beta,f)$. The second equation is obtained from the first one by M\"obius inversion on the poset of uniform block permutations, see Theorem \ref{thm:Mobius-inversion}.
\end{proof}

We are now in position to state and prove the main result of this section.

\begin{theorem} \label{theorem:E-u}
Let $\mathcal{D}$ be a diagram containing $n$ $u$-boxes and $n$ $\overbar u$-boxes, where $u \in \mathbb{T}^d$ is a random vector consisting of i.i.d.~random uniform phases. Then, the following holds true
\begin{equation}
\mathbb{E}_u \mathcal{D} = \sum_{(\alpha, \beta, f) \in \mathcal{UBP}_n} \mathcal{D}_{(\alpha, \beta, f)} \Cf_{\mathcal U}(\alpha, \beta, f),
\end{equation}
where the combinatorial coefficients $\Cf_{\mathcal U}$ were introduced in Definition \ref{def:Cfu}. In other words, the expectation value of $\mathcal{D}$ is a sum of pairings of $\mathcal{D}$ according to all possible uniform block permutations, weighted by the combinatorial factors $\Cf_{\mathcal U}$. 
\end{theorem}

\begin{proof}
By combining the results from Lemmas \ref{lem:Eu-sum-injectve-diagrams} and \ref{lem:Dabf-vs-Dabf-injective}, we get 
\begin{align*}
\mathbb{E}_u \mathcal{D} &= \sum_{(\alpha, \beta, f)\in \mathcal{UBP}_n} \mathcal{D}^{\neq}_{(\alpha, \beta, f)}  \\
&= \sum_{(\alpha, \beta, f)\in \mathcal{UBP}_n} \left[\sum_{(\alpha', \beta', f') \geq (\alpha, \beta, f)} \mathcal{D}_{(\alpha', \beta', f')} \mu_{\mathcal{U}} [(\alpha,\beta,f),(\alpha',\beta',f')] \right] \\
&= \sum_{(\alpha', \beta', f')\in \mathcal{UBP}_n} \mathcal{D}_{(\alpha', \beta', f')} \left[\sum_{(\alpha, \beta, f) \leq (\alpha', \beta', f')} \mu_{\mathcal{U}} [(\alpha,\beta,f),(\alpha',\beta',f')] \right] \\
&= \sum_{(\alpha', \beta', f') \in \mathcal{UBP}_n} \mathcal{D}_{(\alpha', \beta', f')} \Cf_{\mathcal U}(\alpha', \beta', f')
\end{align*}
\end{proof}

\begin{remark}
It is instructive to compare the two formulas from Lemma \ref{lem:Eu-sum-injectve-diagrams} and from Theorem \ref{theorem:E-u}. The former one looks simpler, since the injective diagrams have no combinatorial weight $\Cf_{\mathcal U}$ associated. However, the difficulty with Lemma \ref{lem:Eu-sum-injectve-diagrams} lies in the fact that the injectivity conditions become very intricate even for moderate values of $n$, and avoiding double counting of diagrams is a tedious task. These inconveniences are absent from the formulation in Theorem \ref{theorem:E-u}; the price to pay is the presence of the weights $\Cf_{\mathcal U}$.
\end{remark}

\begin{remark}\label{rem:multiple-integrals-u}
If a diagram $\mathcal{D}$ contains several independent vectors $u_1, u_2, \ldots, u_p$--- each distributed according to Definition \ref{def:rv-u}--- then the expectation $\mathbb{E}_{u_1}\mathbb{E}_{u_2}\ldots \mathbb{E}_{u_p} \mathcal{D}$ can be computed by simultaneous applications of Theorem \ref{theorem:E-u} for each independent vector. More precisely, if $\mathcal{D}$ contains $n_i$ $u_i$-boxes and $n_i$ $\overbar{u_i}$-boxes for $i=1,2,\ldots ,p$, then
\begin{align*}
    \mathbb{E}_{u_1}\mathbb{E}_{u_2}\ldots \mathbb{E}_{u_p} \mathcal{D} &= 
    \mathbb{E}_{u_1}\left(\mathbb{E}_{u_2}\left(\ldots \left(\mathbb{E}_{u_p}\mathcal{D}\right)\ldots\right)\right) \\
    &= \mathlarger{\mathlarger{\sum}}_{ \left\{  \substack{ (\alpha,\beta,f)_1 \in \mathcal{UBP}_{n_1} \\ . \\ . \\ (\alpha,\beta,f)_p \in \mathcal{UBP}_{n_p} }    \right\}} \left\{\prod_{i=1}^p \Cf_\mathcal{U}((\alpha,\beta,f)_i) \right\} \mathcal{D}_{(\alpha,\beta,f)_1, (\alpha,\beta,f)_2, \ldots, (\alpha,\beta,f)_p }
\end{align*}
where the diagrams $\mathcal{D}_{(\alpha,\beta,f)_1, (\alpha,\beta,f)_2, \ldots, (\alpha,\beta,f)_p }$ are constructed as in Definition \ref{def:D-UBP}, with each UBP $(\alpha,\beta,f)_i$ deciding the pairings for the corresponding set of $u_i$, $\overbar{u_i}$ boxes, for $i=1,2,\ldots ,p$.
\end{remark}

\begin{remark}
In Theorem \ref{theorem:E-u}, we have given a recipe for computing averages with respect to random variables having distribution described in Definition \ref{def:rv-u}. This continuous probability distribution (i.i.d.~random phases) can be replaced by discrete probability measures, in the spirit of \emph{quantum $t$-designs} \cite{delsarte1991spherical,renes2004symmetric,dankert2009exact}. In \cite[Appendix A]{nakata2013diagonal} it is shown that random vectors $w \in \mathbb C^d$ having i.i.d.~entries $w_i$ with discrete distribution
$$\forall k \in \{0,1,2,\ldots, n\}, \qquad \mathbb P(w_i = \omega_{n+1}^k) = \frac{1}{n+1},$$
where $\omega_{n+1} = \exp(2\pi \mathrm{i}/(n+1))$, form a \emph{diagonal unitary $n$-design}. In other words, the conclusion of Theorem \ref{theorem:E-u} still holds when replacing the continuous random variable $u$ with the discrete random variable $w$.
\end{remark}

\bigskip

We now give explicit diagrams that one obtains when applying Theorem \ref{theorem:E-u} in the cases $n=1,2$; for the lengthier case $n=3$, see Appendix \ref{sec:app-n-3}. 

\begin{example}\label{eg:complex-n-1}
In the case $n=1$, there a single UBP having coefficient $\Cf_{\mathcal U}=1$, and we obtain: 
\smallskip
\begin{center}
\includegraphics{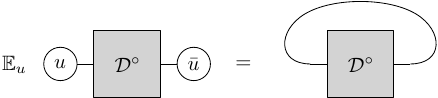}
\end{center}
\smallskip
\end{example}

\begin{example}\label{eg:complex-n-2}
In the case $n=2$, the expectation value is a sum of three diagrams:
\smallskip
\begin{center}
\includegraphics{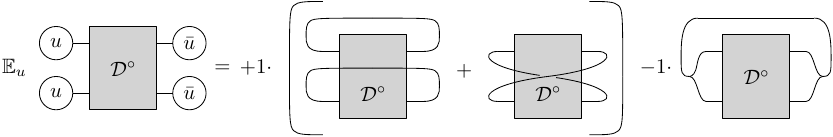}
\end{center}
\smallskip
The three diagrams on the right hand side of the equation above correspond to the following UBPs: 
${\tiny\left(\!\!\begin{tabular}{c|c}
        1 & 2\\
        1 & 2
    \end{tabular}\!\!\right)}$, 
${\tiny\left(\!\!\begin{tabular}{c|c}
        1 & 2\\
        2 & 1
    \end{tabular}\!\!\right)}$,
and, respectively,
${\tiny\left(\!\!\begin{tabular}{c}
        1 2\\
        1 2
    \end{tabular}\!\!\right)}$. The corresponding coefficients are $+1, +1, -1$. 
\smallskip
\end{example}

\bigskip    
    
Finally, let us present an explicit use of Theorem \ref{theorem:E-u}. We shall consider the following problem: given a matrix $X \in \mathcal M_d(\mathbb C)$, compute the expectation value $\mathbb E_U UXU^*$, where $U \in \mathcal U_d$ is a random diagonal unitary matrix having i.i.d.~complex phases on the diagonal. More precisely, $U = \operatorname{diag}(u)$, where $u$ is the random vector from Definition \ref{def:rv-u}, see Figure \ref{fig:E-UXUstar}, left panel. The diagram corresponding to $UXU^*$ is depicted in the top-right panel of Figure \ref{fig:E-UXUstar}; in coordinates, the diagram corresponds to the matrix $\tilde{X}$ with entries $\tilde{X}_{ij} = u_i X_{ij} \overbar {u_j}$. The result of the computation is represented in the bottom-right panel: $\mathbb E_U UXU^* = \operatorname{diag}(X)$.
\begin{figure}[hbt!]
    \centering
    \includegraphics[align=c]{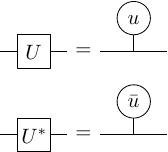} \qquad\qquad  \includegraphics[align=c]{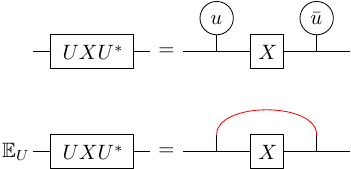}
    \caption{Left: the diagram for a random diagonal unitary matrix $U$ expressed in terms of its diagonal vector $u$, as well as its conjugate. Right, top: the diagram for $UXU^*$, where $X$ is a fixed square matrix. Right, bottom: applying Theorem \ref{theorem:E-u} consists in erasing the $u$ and the $\overbar u$ box and connecting the dangling wires with the unique UBP at $n=1$ (red wire); the result is the diagram corresponding to $\operatorname{diag}(X)$, see also Figure \ref{fig:diagonal}.}
    \label{fig:E-UXUstar}
\end{figure}
    
\section{Graphical integration formula --- real signs}\label{sec:graphical-integration-R}

In this section we shall provide a graphical method to compute the average value of diagrams containing boxes corresponding to random independent uniform $\pm 1$ signs. The content of this section will mirror perfectly that of the previous one, which can be seen as the complex version of the real case discussed here. For this reason, we shall leave many proofs and details to the reader, focusing on examples and on the differences with Section \ref{sec:graphical-integration-C}. 

We shall consider in this section random vectors $s \in \{\pm 1\}^d$, having independent and identically distributed coordinates.

\begin{definition}\label{def:rv-s}
A \emph{random sign} is a random variable $r \in \{\pm 1\}$, with $P(r = -1) =\mathbb P(r = +1) = 1/2 $. A \emph{random sign vector} is a random variable $s\in \{\pm 1\}^d$ having independent and identically distributed coordinates  
\begin{equation}\label{eq:def-rv-s}
    s=(s_1, \ldots, s_d) \text{ where } s_k \in \{\pm1\} \text{ are i.i.d. with } \mathbb P(s_i = -1) =\mathbb P(s_i = +1) = \frac 1 2 
\end{equation}
\end{definition}

Note that there is no conjugate object in this case, so we shall consider diagrams that can be written as a set of $m$ $s$-boxes, see Figure \ref{fig:example-E-s}, left panel. 

\begin{figure}[hbt!]
    \centering
    \includegraphics[align=c]{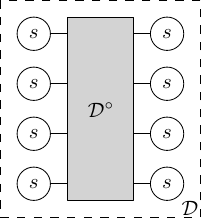} \qquad\qquad \includegraphics[align=c]{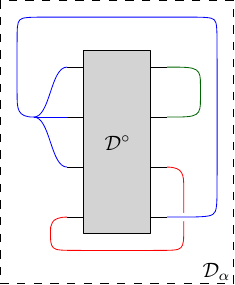} \qquad\qquad \includegraphics[align=c]{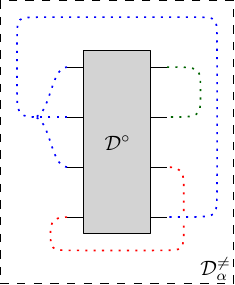}
    \caption{Left: a diagram $\mathcal D$ containing $m=2n=8$ $s$-boxes. Center: the associated diagram $\mathcal D_\alpha$ corresponding to the even partition $\alpha = (\textcolor{blue}{1 2 3 8} | \textcolor{red}{4 7} | \textcolor{darkgreen}{5 6}  )$. Right: the injective diagram $\mathcal D_{\alpha}^{\neq}$ corresponding to the same even partition. The indices corresponding to the dotted lines must be different.}
    \label{fig:example-E-s}
\end{figure}

We start with the analogue of Lemma~\ref{lemma:n neq m = 0}, which will eliminate the trivial cases of diagrams having an odd number of $s$-boxes.

\begin{lemma}
Let $\mathcal{D}$ be a diagram containing an odd number of $s$-boxes. Then, $\mathbb{E}_s \mathcal{D} = 0$.
\end{lemma}
\begin{proof}
    The proof is similar to that of Lemma \ref{lemma:n neq m = 0}: replacing the random vector $s$ by $-s$ (which has the same distribution), we obtain $\mathbb{E}_s \mathcal{D} = -\mathbb{E}_s \mathcal{D}$, proving the claim.
\end{proof}

Given the lemma above, we focus our attention on diagrams having an even number ($2n$) of $s$-boxes. The main result of this section, Theorem \ref{theorem:E-s}, expresses the expectation value of such diagrams $\mathcal D$ as weighted sums of modified diagrams $\mathcal D_\alpha$, indexed over even partitions $\alpha \in \Pi^{(2)}_{2n}$. We introduce the modified diagrams $\mathcal D_\alpha$ (resp.~the injective diagrams $\mathcal D_\alpha^{\neq}$) in a similar way as in Definition \ref{def:D-UBP} (resp.~Definition \ref{def:D-UBP-injective}): we erase the $s$-boxes and then use the even partition $\alpha$ to connect the dangling edges; in the injective case, we require that indices corresponding to different blocks of $\alpha$ should be different, see the center and right panels in Figure \ref{fig:example-E-s} for illustration. The combinatorics and the relation between the injective and non-injective diagrams are governed by the poset of even partitions, with the corresponding M\"obius and combinatorial functions being $\mu_\Pi$ and $\Cf_\Pi$ repsectively, see Section \ref{sec:combinatorics}. We state without proof the two technical combinatorial lemmas which mirror Lemmas \ref{lem:Eu-sum-injectve-diagrams} and \ref{lem:Dabf-vs-Dabf-injective} from the previous section.

\begin{lemma}
Let $\mathcal D$ be a diagram containing $2n$ $s$-boxes. Then,
\begin{equation}
\mathbb{E}_s \mathcal{D} = \sum_{\alpha\in \Pi^{(2)}_{2n}} \mathcal{D}^{\neq}_{\alpha}.
\end{equation}
\end{lemma}

\begin{lemma}
Given a diagram $\mathcal{D}$ and an even partition $\alpha \in \Pi_{2n}^{(2)}$, the following relations hold:
\begin{align}
    \mathcal{D}_\alpha &= \sum_{\alpha' \geq \alpha} \mathcal{D}^{\neq}_{\alpha'}\\
    \mathcal{D}^{\neq}_\alpha &=\sum_{\alpha' \geq \alpha} \mathcal{D}_{\alpha'} \mu_\Pi(\alpha,\alpha').
\end{align}
\end{lemma}

We state now the main result, leaving the proof to the reader. 

\begin{theorem} \label{theorem:E-s}
Let $\mathcal{D}$ be a diagram containing $2n$ $s$-boxes, where $s \in \{\pm 1\}^d$ is a random vector consisting of i.i.d.~random uniform signs. Then, the following holds true
\begin{equation*}
\mathbb{E}_s \mathcal{D} = \sum_{\alpha \in \Pi^{(2)}_{2n}} \mathcal{D}_\alpha \Cf_\Pi(\alpha),
\end{equation*}
where the combinatorial coefficients $\Cf_{\Pi}$ were introduced in Definition \ref{def:Cfpi}. In other words, the expectation value of $\mathcal{D}$ is a sum of pairings of $\mathcal{D}$ according to all possible even partitions, weighted by the combinatorial factors $\Cf_\Pi$. 
\end{theorem}

\begin{remark}
If a diagram $\mathcal{D}$ contains several independent vectors $s_1, s_2, \ldots, s_p$--- each distributed according to Definition \ref{def:rv-s}--- then the expectation $\mathbb{E}_{s_1}\mathbb{E}_{s_2}\ldots \mathbb{E}_{s_p} \mathcal{D}$ can be computed by simultaneous applications of Theorem \ref{theorem:E-s} for each independent vector. More precisely, if $\mathcal{D}$ contains $2n_i$ $s_i$-boxes for $i=1,2,\ldots ,p$. Then
\begin{align*}
    \mathbb{E}_{s_1}\mathbb{E}_{s_2}\ldots \mathbb{E}_{s_p} \mathcal{D} &= 
    \mathbb{E}_{s_1}\left(\mathbb{E}_{s_2}\left(\ldots \left(\mathbb{E}_{s_p}\mathcal{D}\right)\ldots\right)\right) \\
    &= \mathlarger{\mathlarger{\sum}}_{ \left\{  \substack{ \alpha_1 \in \Pi^{(2)}_{2n_1} \\ . \\ . \\ \alpha_p \in \Pi^{(2)}_{2n_p} }    \right\}} \left\{\prod_{i=1}^p \Cf_\Pi(\alpha_i) \right\} \mathcal{D}_{\alpha_1, \alpha_2, \ldots, \alpha_p }
\end{align*}
where the diagrams $\mathcal{D}_{\alpha_1, \alpha_2, \ldots, \alpha_p }$ are constructed as is explained in Figure \ref{fig:example-E-s}, with each even partition $\alpha_i$ in $\Pi^{(2)}_{2n_i}$ deciding the pairings for the corresponding set of $s_i$-boxes, for $i=1,2,\ldots ,p$.
\end{remark}

\bigskip

We now provide explicit graphical expansions which are obtained after applying Theorem \ref{theorem:E-s} for the cases when $n=1$ and $2$.

\begin{example}\label{eg:real-n-1}
In the case $n=1$, there is a single even partition having coefficient $\Cf_\Pi=1$, and we obtain: 
\smallskip
\begin{center}
\includegraphics{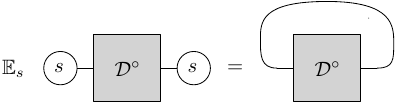}
\end{center}
\smallskip
\end{example}

\begin{example}\label{eg:real-n-2}
In the case $n=2$ there are 4 diagrams, as follows:
\smallskip
\begin{center}
\includegraphics{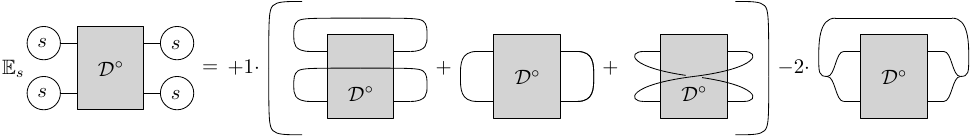}
\end{center}
\smallskip
The first three diagrams correspond to even partitions having two blocks and have coefficient $\Cf_\Pi = 1$, while the last one corresponds to the maximal partition $\mathbb 1_4$ and has coefficient $\Cf_\Pi = -2$.

\end{example}

\bigskip

As a concrete application of the graphical integration formula from Theorem \ref{theorem:E-s}, one can consider the computation of the average $\mathbb E_O OXO^*$, where $O$ is a random diagonal orthogonal matrix, having i.i.d.~signs on the diagonal. As it was the case in Section \ref{sec:graphical-integration-C}, the diagonal random matrix $O$ can be represented in terms of a sign vector $s$, and the final result is identical (see Figure \ref{fig:E-UXUstar}): $\mathbb E_O OXO^* = \operatorname{diag}(X)$.

\section{Local diagonal unitary invariant matrices} \label{sec:LDUI}
In the coming few sections, we explore certain applications of the graphical integration formulas introduced in Theorems \ref{theorem:E-u} and \ref{theorem:E-s}, which are relevant to the theory of Quantum Information. 
We focus in this section and the following one, on three particular families of bipartite matrices introduced in \cite{chruscinski2006class,johnston2019pairwise}: the \emph{local diagonal unitary invariant} (LDUI) matrices, the \emph{conjugate local diagonal unitary invariant} (CLDUI) matrices, and the \emph{local diagonal orthogonal invariant} (LDOI) matrices. These matrices have been considered as potential counter-examples to the \emph{absolutely separable vs.~absolutely PPT conjecture}, see \cite{kus2001geometry,arunachalam2015absolute,IQOQIseparability}. Our goal in this section is to recast some of the results from \cite{johnston2019pairwise} regarding CLDUI matrices in the graphical language we have introduced, as well as to develop the theory of LDUI matrices. The more general case of LDOI matrices will form the subject matter of the next section.

Consider a matrix $X \in \mathcal{M}_d(\mathbb{C}) \otimes \mathcal{M}_d(\mathbb{C})$. We are interested in calculating the expectations $\mathbb{E}_U [(U \otimes  U^*) X (U^* \otimes U)]$ and $\mathbb{E}_U [(U \otimes  U) X (U^* \otimes U^*)]$, where $U \in \mathcal{U}_d$ is a random diagonal unitary matrix $U= \text{diag}(u)$, with the random phase vector $u\in \mathbb{T}^d$ distributed according to Definition \ref{def:rv-u}.

Diagrammatically, the first step is to mould the given tensors in order to make them more tractable (see the discussion following Figure \ref{fig:separate-uubar-D} in Section \ref{sec:graphical-integration-C}). This is done in Figure \ref{fig:UU-X-UU}, which immediately brings the diagrams to a familiar form (Example \ref{eg:complex-n-2}). A simple application of the $n=2$ case of Theorem \ref{theorem:E-u} then gives the desired expectations (see Figure \ref{fig:E-UU-X-UU} below).

\begin{figure}[H]
\centering
\vspace{0.3cm}
\includegraphics[scale=1]{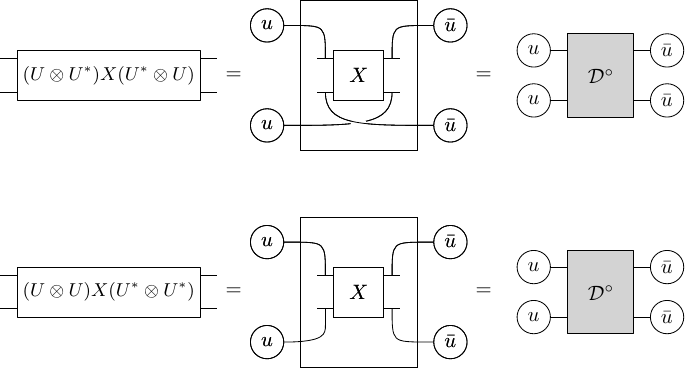}
\caption{The diagrams for $(U \otimes  U^*) X (U^* \otimes U)$ and $(U \otimes  U) X (U^* \otimes U^*)$ (Top and Bottom), reshaped in order to bring them to the standard form.}
\label{fig:UU-X-UU}
\end{figure}

\begin{figure}[H]
    \centering
    \includegraphics[scale=0.9]{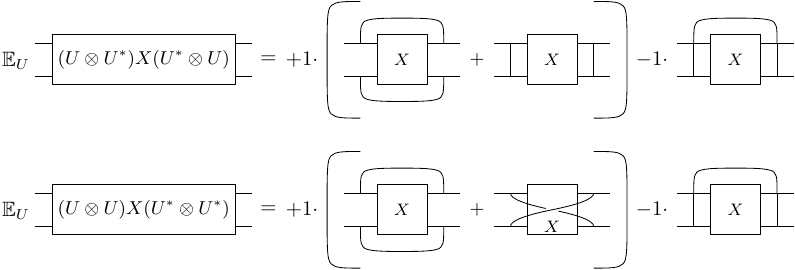}
    \caption{Expectation value of the diagrams $(U \otimes U^*) X (U^* \otimes U)$ and $(U \otimes U)X(U^* \otimes U^*)$ (Top and Bottom) obtained through a simple application of Theorem \ref{theorem:E-u} for the n=2 case. } 
    \label{fig:E-UU-X-UU}
\end{figure}

We further massage these expressions by introducing the following matrices.

\begin{definition} \label{def:X-ABC}
Given $X \in \mathcal{M}_d(\mathbb{C}) \otimes \mathcal{M}_d(\mathbb{C})$, we define the matrices $A,B,C \in \mathcal{M}_d(\mathbb{C})$ as in Figure~\ref{fig:X-ABC}.

\begin{figure}[hbt!] 
\centering
\includegraphics[scale=1]{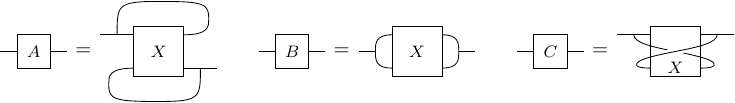}
\caption{The matrices $A,B,C \in \mathcal{M}_d(\mathbb{C})$ corresponding to the bipartite matrix $X \in \mathcal{M}_d(\mathbb{C}) \otimes \mathcal{M}_d(\mathbb{C})$.}
\label{fig:X-ABC}
\end{figure} 

\end{definition} 

Using the matrices $A,B$ and $C$, Figure \ref{fig:E-UU-X-UU} can be redrawn in the form of Figure \ref{fig:E-UU-X-UU-ABC}.

\begin{figure}[hbt!]
    \centering
    \includegraphics{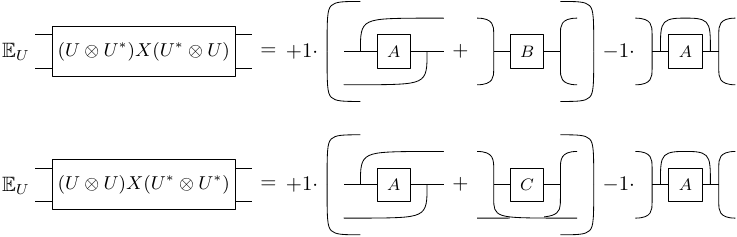}
    \caption{Expectation value of the diagrams $(U \otimes U^*) X (U^* \otimes U)$ and $(U \otimes U)X(U^* \otimes U^*)$ (Top and Bottom) expressed in terms of the associated $A,B,C$ matrices introduced in Definition \ref{def:X-ABC} and Figuure~\ref{fig:X-ABC}.}
    \label{fig:E-UU-X-UU-ABC}
\end{figure}

\begin{remark}\label{remark:X-ABC-diag}
From Definition \ref{def:X-ABC}, it is evident that the diagonal entries of $A,B$ and $C$ are equal. Thus, the matrix $A$ in the last term of the expectation value expressions in Figure \ref{fig:E-UU-X-UU-ABC} can just as well be replaced by $B$ or $C$.
\end{remark}

Our next task is to investigate the properties of bipartite matrices $X$  which stay invariant under the operations $X \mapsto \mathbb{E}_U [(U \otimes  U^*) X (U^* \otimes U)] $ and $X \mapsto \mathbb{E}_U [(U \otimes  U) X (U^* \otimes U^*)]$. We begin with the relevant definition, see \cite[Definition 5.1]{johnston2019pairwise}.

\begin{definition}[LDUI/CLDUI matrices]\label{def:LDUI/CLDUI}
A matrix $X \in \mathcal{M}_d(\mathbb{C}) \otimes \mathcal{M}_d(\mathbb{C})$ is said to be \emph{Local diagonal unitary invariant (LDUI)} (resp.~\emph{Conjugate local diagonal unitary invariant (CLDUI)}) if $(U \otimes  U) X (U^* \otimes U^*) = X$ (resp.~$(U \otimes  U^*) X (U^* \otimes U) = X$) for all diagonal unitary matrices $U\in \mathcal{U}_d$.
\end{definition}

\begin{proposition} \label{prop:LDUI/CLDUI-AC/AB}
A matrix $X \in \mathcal{M}_d(\mathbb{C}) \otimes \mathcal{M}_d(\mathbb{C})$ is LDUI (resp. CLDUI) if and only if the mappings $X \mapsto \mathbb{E}_U [(U \otimes  U) X (U^* \otimes U^*)]$ (resp. $X \mapsto \mathbb{E}_U [(U \otimes  U^*) X (U^* \otimes U)]$ leave $X$ invariant. Hence, the set of bipartite LDUI (resp.~CLDUI) matrices is in bijection with the set of matrix pairs $(A,C)$ (resp.~$(A,B)$) in $\mathcal{M}_d(\mathbb{C}) \times \mathcal{M}_d(\mathbb{C})$ satisfying $\operatorname{diag}(A)=\operatorname{diag}(C)$ (resp.~$\operatorname{diag}(A)=\operatorname{diag}(B)$).
\end{proposition}

\begin{proof} 
The ``only if" direction is trivial to show, since if $X \in \mathcal{M}_d(\mathbb{C}) \otimes \mathcal{M}_d(\mathbb{C})$ is LDUI (resp. CLDUI), then it is clear from Definition~\ref{def:LDUI/CLDUI} that it stays invariant under the given mappings. Conversely, if we assume the invariance of $X$ under the given mappings, then for any fixed diagonal unitary matrix $V\in \mathcal{U}_d$, we have 
\begin{align}
    (V \otimes  V) X (V^* \otimes V^*) &= \mathbb{E}_U [(UV \otimes  UV) X ( V^*U^* \otimes V^*U^*)] = X \nonumber \\
    (V \otimes  V^*) X (V^* \otimes V) &= \mathbb{E}_U [(UV \otimes  V^*U^*) X ( V^*U^* \otimes UV)] = X
\end{align}
where the equalities follow from the distributional equivalence $U \stackrel{dist}{=} UV$. It is now evident from Figure~\ref{fig:E-UU-X-UU-ABC} and Remark~\ref{remark:X-ABC-diag} that to each LDUI (resp. CLDUI) matrix $X$ is associated a unique matrix pair $(A,C)$ (resp. $(A,B)$) such that $\operatorname{diag}(A) = \operatorname{diag}(C)$ (resp. $\operatorname{diag}(A) = \operatorname{diag}(B)$).
\end{proof}

Proposition \ref{prop:LDUI/CLDUI-AC/AB} allows us to interpret the operation $X \mapsto \mathbb{E}_U [(U \otimes  U) X (U^* \otimes U^*)]$ (resp. $X \mapsto \mathbb{E}_U [(U \otimes  U^*) X (U^* \otimes U)])$ as an orthogonal projection from the $d^4$-dimensional $\mathbb{C}$-Hilbert space $\mathcal{M}_d(\mathbb{C}) \otimes \mathcal{M}_d(\mathbb{C})$ onto the smaller $(2d^2 - d)$-dimensional subspace of LDUI (resp. CLDUI) matrices.

We now investigate the conditions which ensure that a given LDUI/CLDUI matrix $X$ is \textit{separable}, i.e., there exists a family of vectors $\{ v_k, w_k \}_{k \in I} \subseteq \mathbb{C}^d$ for a finite index set $I$, such that 
\begin{equation}
    X = \sum_k v_k v^*_k \otimes w_k w^*_k
\end{equation}
\begin{figure}[hbt!]
    \centering
    \includegraphics{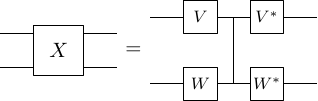}
    \caption{Separability condition for a matrix $X \in \mathcal{M}_d(\mathbb{C}) \otimes \mathcal{M}_d(\mathbb{C})$.}
    \label{fig:X-sep}
\end{figure}

Graphically, this decomposition is equivalent to the diagram given in Figure \ref{fig:X-sep}, where $V$ and $W$ are matrices in $\mathcal{M}_{d \times \vert I \vert}(\mathbb{C})$ with columns given by the vectors $\{ {v_k} \}_{k\in I}$ and $\{ {w_k} \}_{k \in I}$ respectively. It is important to note that the size of the index set $I$ (equivalently, the dimension of vector space corresponding to the 4-valent wire in the center of Figure \ref{fig:X-sep}) corresponds to the \emph{length of the PCP decomposition} from \cite[Section 5.1]{johnston2019pairwise}. Clearly, if $X$ admits the given decomposition, it is trivial to write down the diagrams for the corresponding matrices $A,B$ and $C$ from Definition \ref{def:X-ABC} (see Figure \ref{fig:X-ABC-sep}). Conversely, if there exist matrices $V,W \in \mathcal{M}_{d \times \vert I \vert}(\mathbb{C})$ such that the matrix pair $(A,C)$ (resp.~$(A,B)$) can be decomposed as in Figure~\ref{fig:X-ABC-sep}, then $X$ in $\mathcal{M}_d(\mathbb{C}) \otimes \mathcal{M}_d(\mathbb{C})$ given by the diagram in Figure \ref{fig:X-sep} is clearly separable (though not necessarily LDUI/CLDUI). The LDUI (resp.~CLDUI) matrix associated with the pair $(A,C)$ (resp.~$(A,B)$) is then obtained from $X$ via the operation $X \mapsto \mathbb{E}_U [(U \otimes  U) X (U^* \otimes U^*)]$ (resp.~$X\mapsto\mathbb{E}_U [(U \otimes  U^*) X (U^* \otimes U)])$ which preserves separability. We state the above discussion more precisely below, using the terminology introduced in \cite[Definition 3.1]{johnston2019pairwise}, generalizing the classical notion of \emph{completely positive matrices} \cite{abraham2003completely} (not to be confused with completely positive maps from operator algebra).


\begin{figure}[hbt!]
    \centering
    \includegraphics{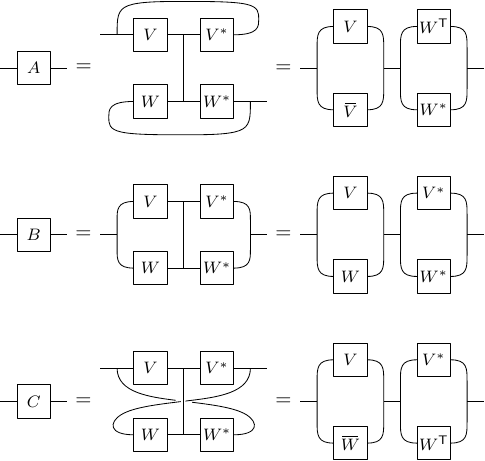}
    \caption{Separability condition for a matrix $X \in \mathcal{M}_d(\mathbb{C}) \otimes \mathcal{M}_d(\mathbb{C})$, expressed in terms of the associated matrices $A,B$ and $C \in \mathcal{M}_d(\mathbb{C})$ }
    \label{fig:X-ABC-sep}
\end{figure}

\begin{definition}[Pairwise completely positive matrices] \label{def:pcp}
A pair $(Z,Z')$ in $\mathcal{M}_d(\mathbb{C}) \times \mathcal{M}_d(\mathbb{C})$ is said to be \emph{pairwise completely positive} if there exist matrices $V,W \in \mathcal{M}_{d,d'}(\mathbb{C})$ (for $d' \in \mathbb{N}$) such that
\begin{equation}
    Z = (V \odot \overbar{V})(W \odot \overbar{W})^* \qquad Z' = (V \odot W) (V \odot W)^*
\end{equation}
\end{definition}

\begin{lemma}[Separability of LDUI/CLDUI matrices]\label{lemma:LDUI/CLDUI-sep}
An LDUI (resp.~CLDUI) matrix $X$ in $\mathcal{M}_d(\mathbb{C}) \otimes \mathcal{M}_d(\mathbb{C})$ is separable if and only if the associated matrix pair $(A,C)$ (resp.~$(A,B)$) in $\mathcal{M}_d(\mathbb{C}) \times \mathcal{M}_d(\mathbb{C})$ is pairwise completely positive.
\end{lemma}

Next, we analyse the constraints required on matrices $A,B$ and $C$ to ensure that the corresponding LDUI/CLDUI matrix $X$ and its partial transpose $X^{\Gamma}=[\operatorname{id} \otimes \operatorname{transp}](X)$ are positive semi-definite.

\begin{lemma} \label{lemma:LDUI-psd/ppt}
Let $X$ in $\mathcal{M}_d(\mathbb{C}) \otimes \mathcal{M}_d(\mathbb{C})$ be an LDUI matrix with the associated matrix pair $(A,C)$ in $\mathcal{M}_d(\mathbb{C}) \times \mathcal{M}_d(\mathbb{C})$. Then, 
\begin{enumerate}
    \item $X$ is self-adjoint if and only if $A$ is real and $C$ is self-adjoint.
    \item $X$ is positive semi-definite if and only if $A$ is entrywise non-negative, $C$ is self-adjoint and and $A_{ij}A_{ji} \geq \vert C_{ij} \vert^2$ for $i,j \in [d]$.
    \item $X^{\Gamma}$ is positive semi-definite if and only if $A$ is entrywise non-negative and $C$ is positive semi-definite. 
\end{enumerate}
\end{lemma}

\begin{proof}
    Part (1) is trivially proved by taking adjoint of the diagram given in the bottom panel of Figure \ref{fig:E-UU-X-UU-ABC}. For part (2), we consider an arbitrary vector $v = \sum_{i,j=1}^d v_{ij} e_i \otimes e_j \in \mathbb{C}^d \otimes \mathbb{C}^d$ and compute
    \begin{align*}
    v^* X v &= \includegraphics[align=c]{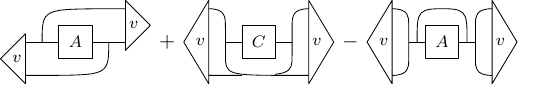} \\
        &= \sum_{i,j=1}^d A_{ij}\vert v_{ij} \vert^2 + \sum_{i,j=1}^d C_{ij}\overbar{v_{ij}}v_{ji} - \sum_{i=1}^d A_{ii} \vert v_{ii} \vert^2 \\
        &= \sum_{i=1}^d C_{ii} \vert v_{ii} \vert^2 + \sum_{i\neq j=1}^d \left(A_{ij}\vert v_{ij} \vert^2 + C_{ij}\overbar{v_{ij}}v_{ji}\right) \\
        &= \sum_{i=1}^d C_{ii} \vert v_{ii} \vert^2 + \mathlarger{\sum}_{i<j=1}^d 
        \begin{array}{c@{}@{}}
         \left(\begin{matrix}
                \overbar{v_{ij}} & \overbar{v_{ji}}
                \end{matrix}\right)  \\
                \hspace{0.01cm}
        \end{array}
        \left(\begin{matrix}
                    A_{ij} & C_{ij}  \\
                    C_{ji} & A_{ji}
                  \end{matrix}\right)
        \left(\begin{matrix}
                v_{ij} \\
                v_{ji}
                \end{matrix}\right)
    \end{align*}

The above expression is non-negative if and only if the conditions in part (2) are met. 

For the third part, we compute (see Figure \ref{fig:E-UU-X-UU-ABC}, bottom panel)
\begin{align*}
v^* X^\Gamma v &= \includegraphics[align=c]{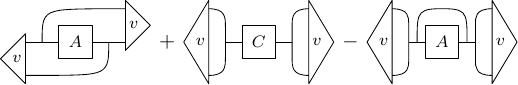}\\
    &= \sum_{i,j=1}^d A_{ij}\vert v_{ij} \vert^2 + \sum_{i,j=1}^d C_{ij}\overbar{v_{ii}}v_{jj} - \sum_{i=1}^d A_{ii} \vert v_{ii} \vert^2 \\
    &= \sum_{i\neq j=1}^d A_{ij} \vert v_{ij} \vert^2+ \operatorname{diag}(v)^* C \operatorname{diag}(v) \\
    &=\includegraphics[align=c]{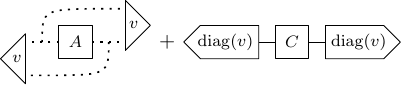}\\    
\end{align*}
where $\operatorname{diag}(v) = \sum_{i=1}^d v_{ii} e_i \in \mathbb{C}^d$ and the dotted wires correspond to an injective diagram (i.e.~the respective indices must be different). The above expression is non-negative if and only if the conditions in part (3) are met. This completes the proof of the lemma.
\end{proof}

Note that the partial transpose of an LDUI matrix gives a CLDUI matrix (and vice-versa):
$$X = \mathbb E_U(U \otimes U^*)X(U^* \otimes U) \iff X^\Gamma = [\mathbb E_U(U \otimes U^*)X(U^* \otimes U)]^\Gamma = \mathbb E_U(U \otimes U)X^\Gamma(U^* \otimes U^*). $$
We can thus immediately write down the analogue of Lemma \ref{lemma:LDUI-psd/ppt} for CLDUI states; this statement is contained in \cite[Theorem 5.2]{johnston2019pairwise}.

\begin{lemma} \label{lemma:CLDUI-psd/ppt}
Let $X$ in $\mathcal{M}_d(\mathbb{C}) \otimes \mathcal{M}_d(\mathbb{C})$ be a CLDUI matrix with the associated matrix pair $(A,B)$ in $\mathcal{M}_d(\mathbb{C}) \times \mathcal{M}_d(\mathbb{C})$. Then, 
\begin{enumerate}
    \item $X$ is self-adjoint if and only if $A$ is real and $B$ is self-adjoint.
    \item $X$ is positive semi-definite if and only if $A$ is entrywise non-negative and $B$ is positive semi-definite.
    \item $X^{\Gamma}$ is positive semi-definite if and only if $A$ is entrywise non-negative, $B$ is self-adjoint and $A_{ij}A_{ji} \geq \vert B_{ij} \vert^2$ for $i,j \in [d]$.
\end{enumerate}
\end{lemma}

Finally, we see how the trace expression for an LDUI/CLDUI matrix $X$ manifests itself as the entrywise sum of the elements of the associated matrix $A$.

\begin{lemma}\label{lemma:LDUI/CLDUI-Tr}
Let $X$ in $\mathcal{M}_d(\mathbb{C}) \otimes \mathcal{M}_d(\mathbb{C})$ be an LDUI (resp.~CLDUI) matrix with the associated matrix pair $(A,C)$ (resp.~$(A,B)$) in $\mathcal{M}_d(\mathbb{C}) \times \mathcal{M}_d(\mathbb{C})$. Then $\Tr(X)=\sum_{i,j=1}^d A_{ij}$.
\end{lemma}

\begin{proof}
    The trace expressions for a CLDUI matrix with matrix pair $(A,B)$ and an LDUI matrix with matrix pair $(A,C)$ are given in Figure \ref{fig:X-ABC-Tr} (Top and Bottom panel respectively). The last two terms in both the expressions cancel on account of Remark \ref{remark:X-ABC-diag} and we are left with just the first term which equals $\sum_{i,j=1}^d A_{ij}$.
    
    \begin{figure}[hbt!]
        \centering
        \includegraphics{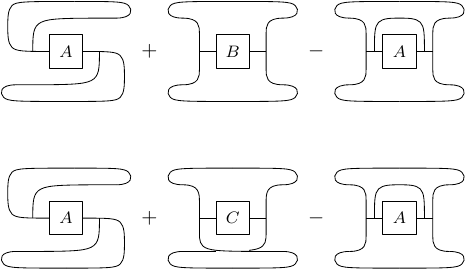}
        \caption{Trace expressions for a CLDUI (top) / LDUI (bottom) matrix expressed in terms of the associated matrices $A,B$ and $C$. The two last diagrams cancel, and one is left in both cases with the first diagram.}
        \label{fig:X-AB-AC-Tr}
    \end{figure}
    
\end{proof}

\section{Local diagonal orthogonal invariant matrices}

Let us discuss now the case of \emph{local diagonal orthogonal invariant} (LDOI) matrices. The analysis here will perfectly mirror the previous one, with the only distinction being the replacement of the random diagonal unitary matrix $U\in \mathcal{U}_d$ with the random diagonal orthogonal matrix $O\in \mathcal{O}_d$, i.e., for $X\in \mathcal{M}_d(\mathbb{C}) \otimes \mathcal{M}_d(\mathbb{C})$, we will be interested in computing the average $\mathbb{E}_O [(O \otimes O)X(O \otimes O)]$, where $O\in \mathcal{O}_d$ is a random diagonal orthogonal matrix $O=\text{diag}(s)$, with the random sign vector $s\in \{\pm 1\}^d$ distributed according to Definition \ref{def:rv-s}. 

\begin{figure}[hbt!]
    \centering
    \includegraphics{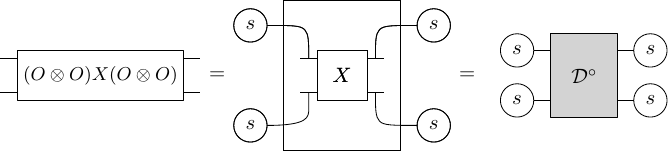}
    \caption{The diagram for $(O \otimes O)X(O \otimes O)$, reshaped in order to bring it to the standard form.}
    \label{fig:OO-X-OO}
\end{figure}

As usual, the initial step is to rearrange the given tensor so as to bring it to the standard form (see Figure \ref{fig:OO-X-OO}), which can be dealt with through a simple application of the $n=2$ case of Theorem \ref{theorem:E-s} (see Example \ref{eg:real-n-2}). Then, with the help of matrices $A,B$ and $C$ in $\mathcal{M}_d(\mathbb{C})$ introduced in Definition \ref{def:X-ABC}, the resulting expectation can be transmuted in the form of Figure \ref{fig:E-OO-X-OO}.

\begin{figure}[hbt!]
    \centering
    \includegraphics{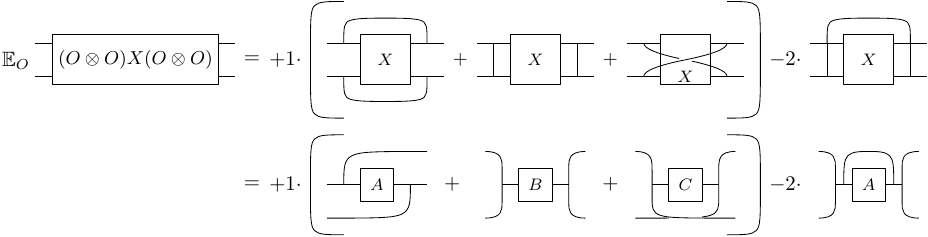}
    \caption{Expectation value of $(O \otimes O)X(O \otimes O)$, expressed in terms of the associated matrices $A,B$ and $C$ from Definition \ref{def:X-ABC}.}
    \label{fig:E-OO-X-OO}
\end{figure}

\begin{definition}[LDOI matrices]\label{def:LDOI}
A matrix $X \in \mathcal{M}_d(\mathbb{C}) \otimes \mathcal{M}_d(\mathbb{C})$ is said to be \emph{Local diagonal orthogonal invariant (LDOI)} if $(O \otimes  O) X (O \otimes O) = X$ for all diagonal orthogonal matrices $O \in \mathcal{O}_d$.
\end{definition}

We now state the equivalent of Proposition~\ref{prop:LDUI/CLDUI-AC/AB} for LDOI matrices, leaving an analogous proof to the good sense of the reader.

\begin{proposition} \label{prop:LDOI-ABC}
A matrix $X \in \mathcal{M}_d(\mathbb{C}) \otimes \mathcal{M}_d(\mathbb{C})$ is LDOI if and only if the mapping $X \mapsto \mathbb{E}_O [(O \otimes  O) X (O \otimes O)]$ leave $X$ invariant. Hence, the set of bipartite LDOI matrices is in bijection with the set of matrix triples $(A,B,C)$ in $\mathcal{M}_d(\mathbb{C}) \times \mathcal{M}_d(\mathbb{C}) \times \mathcal{M}_d(\mathbb{C})$ satisfying $\operatorname{diag}(A)=\operatorname{diag}(B)=\operatorname{diag}(C)$.
\end{proposition}

Thus, we can view the operation $X \mapsto \mathbb{E}_O [(O \otimes O)X(O \otimes O)]$ as an orthogonal projection from the bigger $d^4$-dimensional $\mathbb{C}$ Hilbert space $\mathcal{M}_d(\mathbb{C}) \otimes \mathcal{M}_d(\mathbb{C})$ onto the smaller $\mathbb{C}$ subspace of LDOI matrices with dimension $3d^2 - 2d$. The explicit matrix structure (upto rearrangement of basis elements) of an LDOI matrix $X$ with the associated triple $(A,B,C)$ splits nicely into the block diagonal form given in Eq.~\eqref{eq:LDOI-block}. For $d=3$, this corresponds to the $3\times 3$ block structure in Eq.~\eqref{eq:LDOI-block-3}.
\begin{equation} \label{eq:LDOI-block}
    X = B \oplus \bigoplus_{i<j} \left( \begin{array}{cc}
        A_{ij} & C_{ij} \\
        C_{ji} & A_{ji}
    \end{array}   \right)
\end{equation}
\begin{equation} \label{eq:LDOI-block-3}
    X =  \left(\mathcode`0=\cdot
\begin{array}{ *{3}{c} | *{3}{c} | *{3}{c} }
   A_{11} & 0 & 0 & 0 & B_{12} & 0 & 0 & 0 & B_{13} \\
   0 & A_{12} & 0 & C_{12} & 0 & 0 & 0 & 0 & 0 \\
   0 & 0 & A_{13} & 0 & 0 & 0 & C_{13} & 0 & 0 \\\hline
   0 & C_{21} & 0 & A_{21} & 0 & 0 & 0 & 0 & 0 \\
   B_{21} & 0 & 0 & 0 & A_{22} & 0 & 0 & 0 & B_{23} \\
   0 & 0 & 0 & 0 & 0 & A_{23} & 0 & C_{23} & 0 \\ \hline
   0 & 0 & C_{31} & 0 & 0 & 0 & A_{31} & 0 & 0 \\
   0 & 0 & 0 & 0 & 0 & C_{32} & 0 & A_{32} & 0 \\
   B_{31} & 0 & 0 & 0 & B_{32} & 0 & 0 & 0 & A_{33} \\
  \end{array}
\right)
\end{equation}

\begin{remark}\label{remark:LDOI>LDUI}
Local diagonal orthogonal invariance generalizes the notion of local diagonal unitary invariance for bipartite matrices in $\mathcal{M}_d(\mathbb{C}) \otimes \mathcal{M}_d(\mathbb{C})$.
This can be seen, for instance, by considering an arbitrary matrix triple $(A,B,C)$ defining an LDOI matrix $X$ and replacing either $B$, or $C$ with its corresponding diagonal counterpart. This yields either an LDUI matrix (if $B$ is replaced), or a CLDUI matrix (if $C$ is replaced), see Figures \ref{fig:E-UU-X-UU-ABC},\ref{fig:E-OO-X-OO} and Eq.~\eqref{eq:LDOI-block}. This shows that the set of LDOI matrices strictly contains the set of LDUI/CLDUI matrices.
\end{remark}

Next, we move on to the separability analysis. If a given LDOI matrix $X$ is separable, we can write down the expressions for the associated matrices $A,B,C$ as we did before for the LDUI/CLDUI matrices, getting the exact same decomposition as in Figure~\ref{fig:X-ABC-sep}. Conversely, if $A,B,C$ admit this decomposition, the separability of the matrix $X$ in $\mathcal{M}_d(\mathbb{C}) \otimes \mathcal{M}_d(\mathbb{C})$ given by the diagram in Figure \ref{fig:X-sep} implies the separability of the desired LDOI matrix via the separability preserving map $X \mapsto \mathbb{E}_O [(O \otimes O)X(O \otimes O)]$. To state this result more precisely, we first introduce a generalization of Definition \ref{def:pcp} to \textit{triplewise completely positive matrices.}

\begin{definition}[Triplewise completely positive matrices] \label{def:tcp}
A matrix triple $(Z,Z',Z'')$ in $\mathcal{M}_d(\mathbb{C}) \times \mathcal{M}_d(\mathbb{C}) \times \mathcal{M}_d(\mathbb{C})$ is said to be \emph{triplewise completely positive} if there exist matrices $V,W \in \mathcal{M}_{d,d'}(\mathbb{C})$ (for $d' \in \mathbb{N}$) such that
\begin{equation}
    Z = (V \odot \overbar{V})(W \odot \overbar{W})^* \qquad Z' = (V \odot W) (V \odot W)^* \qquad Z'' = (V \odot \overbar{W})(V \odot \overbar{W})^*
\end{equation}
\end{definition}

\begin{lemma}[Separability of LDOI matrices] \label{lemma:LDOI-sep}
An LDOI matrix $X$ in $\mathcal{M}_d(\mathbb{C}) \otimes \mathcal{M}_d(\mathbb{C})$ is separable if and only if the associated matrix triple $(A,B,C)$ in $\mathcal{M}_d(\mathbb{C}) \times \mathcal{M}_d(\mathbb{C}) \times \mathcal{M}_d(\mathbb{C})$ is triplewise completely positive.
\end{lemma}

Some further properties of triplewise completely positive matrices are presented in Appendix \ref{sec:app-tcp}. We now state the analogue of Lemma \ref{lemma:LDUI-psd/ppt} and \ref{lemma:CLDUI-psd/ppt} for LDOI matrices.

\begin{lemma} \label{lemma:LDOI-psd-ppt}
Let $X$ in $\mathcal{M}_d(\mathbb{C}) \otimes \mathcal{M}_d(\mathbb{C})$ be an LDOI matrix with the associated matrix triple $(A,B,C)$ in $\mathcal{M}_d(\mathbb{C}) \times \mathcal{M}_d(\mathbb{C}) \times \mathcal{M}_d(\mathbb{C})$. Then 
\begin{enumerate}
    \item $X$ is self-adjoint if and only if $A$ is real and $B,C$ are self-adjoint.
    \item $X$ is positive semi-definite if and only if $A$ is entrywise non-negative, $B$ is positive semi-definite, $C$ is self-adjoint and $A_{ij}A_{ji} \geq \vert C_{ij} \vert^2$ for $i,j \in [d]$.
    \item $X^\Gamma$ is positive semi-definite if and only if $A$ is entry wise non-negative, $B$ is self-adjoint, $C$ is positive semi-definite, and $A_{ij}A_{ji} \geq \vert B_{ij} \vert^2$ for $i,j \in [d]$.
\end{enumerate}
\end{lemma}

\begin{proof}
     Part (1) is trivially proved by taking adjoint of the expression in Figure~\ref{fig:E-OO-X-OO}. \\
     For part (2), we compute $v^* X v $ for an arbitrary $v = \sum_{i,j=1}^d v_{ij} e_i \otimes e_j \in \mathbb{C}^d \otimes \mathbb{C}^d$.
\begin{align*}
    v^* Xv  &= \includegraphics[align=c]{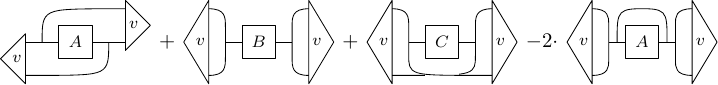} \\
    &= \sum_{i,j=1}^d A_{ij}\vert v_{ij} \vert^2 + \sum_{i,j=1}^d B_{ij} \overbar{v_{ii}}v_{jj} + \sum_{i,j=1}^d C_{ij}\overbar{v_{ij}}v_{ji} -2 \sum_{i=1}^d A_{ii} \vert v_{ii} \vert^2 \\
    &= \sum_{i,j=1}^d B_{ij} \overbar{v_{ii}}v_{jj} + \sum_{i\neq j=1}^d \left(A_{ij}\vert v_{ij} \vert^2 + C_{ij}\overbar{v_{ij}}v_{ji}\right) \\ 
    &= \operatorname{diag}(v)^* B \operatorname{diag}(v) + \mathlarger{\sum}_{i<j=1}^d 
        \begin{array}{c@{}@{}}
         \left(\begin{matrix}
                \overbar{v_{ij}} & \overbar{v_{ji}}
                \end{matrix}\right)  \\
                \hspace{0.01cm}
        \end{array}
        \left(\begin{matrix}
                    A_{ij} & C_{ij}  \\
                    C_{ji} & A_{ji}
                  \end{matrix}\right)
        \left(\begin{matrix}
                v_{ij} \\
                v_{ji}
                \end{matrix}\right)
\end{align*} where $\operatorname{diag}(v) = \sum_{i=1}^d v_{ii} e_i \in \mathbb{C}^d$. This expression is non-negative if and only if the conditions in part (2) are fulfilled. For part (3), we note that the partial transposition of an LDOI matrix simply interchanges the matrices $B$ and $C$ with each other, thus instantly yielding the desired result. 
\end{proof}

The trace expression for an LDOI matrix $X$ in $\mathcal{M}_d(\mathbb{C}) \otimes \mathcal{M}_d(\mathbb{C})$ is identical to its counterpart for an LDUI/CLDUI matrix. 

\begin{lemma}\label{lemma:LDOI-Tr}
Let $X$ in $\mathcal{M}_d(\mathbb{C}) \otimes \mathcal{M}_d(\mathbb{C})$ be an LDOI matrix with the associated matrix triple $(A,B,C)$ in $\mathcal{M}_d(\mathbb{C}) \times \mathcal{M}_d(\mathbb{C}) \times \mathcal{M}_d(\mathbb{C})$. Then $\Tr(X)=\sum_{i,j=1}^d A_{ij}$.
\end{lemma}

\begin{proof}
    We can write down the trace expression for an LDOI matrix $X$ with the matrix triple $(A,B,C)$ as follows. 
    \begin{figure}[hbt!]
        \centering
        \includegraphics{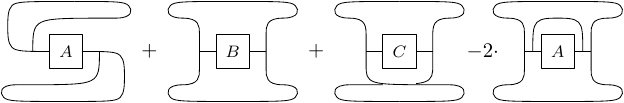}
        \caption{Trace expressions for a CLDUI (top) / LDUI (bottom) matrix expressed in terms of the associated matrices $A,B$ and $C$. The two last diagrams cancel, and one is left in both cases with the first diagram.}
        \label{fig:X-ABC-Tr}
    \end{figure}
    The last three terms in this expression cancel on account of Remark \ref{remark:X-ABC-diag} and we are left with just the first term which equals $\sum_{i,j=1}^d A_{ij}$.
\end{proof}

\section{Diagonal twirling of linear maps between matrix algebras}\label{sec:twirl}

We consider in this section three different \emph{twirling} operations for linear maps between matrix algebras, obtained by averaging with respect to two \emph{independent} random diagonal unitary matrices. The second case discussed here has already been considered in \cite{yu2016bounds,harris2018schur} in relation to the theory of Schur multipliers and mixed unitary matrices; the other two cases are new. More general classes of linear maps, obtained by averaging with respect to one diagonal unitary matrix, shall be considered in a future work \cite{singh2020diagonal}. 

We now introduce three twirling operations, acting on linear maps $\Phi: \mathcal M_d(\mathbb C) \to \mathcal M_d(\mathbb C)$:
\begin{align}
    \label{eq:twirl-1}\mathcal T_\equal(\Phi)(X) &= \int_{U,V} U \Phi(V^*XV)U^* \mathrm{d}U \mathrm{d}V\\
    \label{eq:twirl-2}\mathcal T_\shortparallel(\Phi)(X) &= \int_{U,V} U \Phi(U^*XV^*)V \mathrm{d}U \mathrm{d}V\\
    \label{eq:twirl-3}\mathcal T_\times(\Phi)(X) &= \int_{U,V} U \Phi(V^*XU^*)V \mathrm{d}U \mathrm{d}V,    
\end{align}
where $\mathrm{d}U$, $\mathrm{d}V$ denote the Haar measure on the group of diagonal unitary matrices. In other words, $U = \operatorname{diag}(u)$ (resp.~$V = \operatorname{diag}(v)$), for $u,v$ independent random variables as in Definition \ref{def:rv-u}. 

The action of the twirling operators $\mathcal T_{\equal, \shortparallel, \times}$ on linear maps shall be analyzed on the level of their Choi-Jamio{\l}kowski matrices, see \cite[Section 2.2.2]{watrous2018theory}:
$$\mathcal M_d(\mathbb C) \otimes \mathcal M_d(\mathbb C) \ni J(\Phi) := \sum_{i,j=1}^d \Phi(e_ie_j^*) \otimes e_ie_j^*,$$
for an orthonormal basis $\{e_i\}_{i=1}^d$ of $\mathbb C^d$. The equation $\Phi(X) = [\operatorname{id} \otimes \operatorname{Tr}](J(\Phi)(\mathbb{I}_d \otimes X^{\mathsf T}))$ retrieves the action of $\Phi$ on an input $X$ in $\mathcal{M}_d(\mathbb{C})$ from its Choi-Jamio{\l}kowski matrix $J(\Phi)$, which is diagrammatically represented in Figure \ref{fig:Choi-V}, left panel. 

\begin{figure}[H]
    \centering
    \includegraphics[align=c]{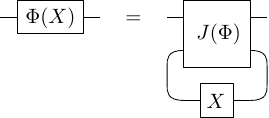}\qquad\qquad\qquad
    \includegraphics[align=c]{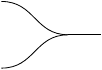}\qquad\qquad\qquad
    \includegraphics[align=c]{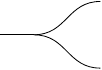}
    \caption{Left: the action of a linear map $\Phi: \mathcal M_d(\mathbb C) \to \mathcal M_d(\mathbb C)$ in terms of its Choi-Jamio{\l}kowski matrix $J(\Phi)$. Center and right: the copy isometry $V$ and its dual $V^*$.}
    \label{fig:Choi-V}
\end{figure}

\begin{figure}
    \centering
    \includegraphics{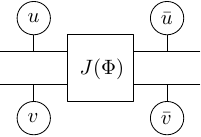}\qquad\qquad
    \includegraphics{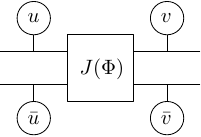}\qquad\qquad
    \includegraphics{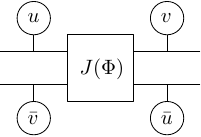}
    \caption{The effect of twirling on a linear map $\Phi: \mathcal M_d(\mathbb C) \to \mathcal M_d(\mathbb C)$ with diagonal unitary operators $U=\operatorname{diag}(u)$ and $V=\operatorname{diag}(v)$. From left to right, the diagrams correspond to the Choi-Jamio{\l}kovski matrices of the integrands in Eqs.~\eqref{eq:twirl-1}, \eqref{eq:twirl-2}, \eqref{eq:twirl-3}.}
    \label{fig:twirl-choi}
\end{figure}

\begin{figure}[htb!]
    \centering
    \includegraphics{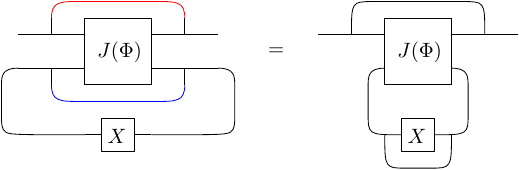}\\ \vspace{.5cm}
    \includegraphics{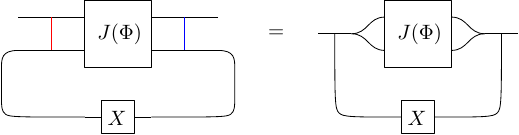}\\ \vspace{.5cm}
    \includegraphics{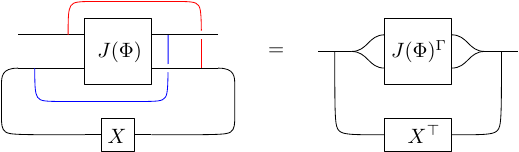}
    \caption{The expectation values of the diagrams from Figure \ref{fig:twirl-choi}, seen as linear maps $\Phi: \mathcal M_d(\mathbb C) \to \mathcal M_d(\mathbb C)$. From top to bottom, we draw the diagrams of the matrices $T_\equal(\Phi)(X)$, $\mathcal T_\shortparallel(\Phi)(X)$, and $\mathcal T_\times(\Phi)(X)$. The averages are computed by applying twice Theorem \ref{theorem:E-u}, once for $\textcolor{red}{\mathbb E_u}$ and a second time for $\textcolor{blue}{\mathbb E_v}$.}
    \label{fig:twirl-final}
\end{figure}

We state now the main result of this section which provides explicit formulas for the ``twirled'' linear maps $\mathcal T_{\equal, \shortparallel, \times}(\Phi)$; the case of the map $\mathcal T_\shortparallel(\Phi)$ has already been considered in \cite[Section V]{yu2016bounds} and \cite[Lemma III.4]{harris2018schur}.

\begin{proposition}
The actions of twirling operations described in \eqref{eq:twirl-1}-\eqref{eq:twirl-3} on linear maps $\Phi: \mathcal M_d(\mathbb C) \to \mathcal M_d(\mathbb C)$ are, respectively: 
\begin{align}
    \label{eq:action-1}\mathcal T_\equal(\Phi)(X) &= \operatorname{diag}(\Phi(\operatorname{diag}(X))\\
    \label{eq:action-2}\mathcal T_\shortparallel(\Phi)(X) &= [V^*J(\Phi)V] \odot X\\
    \label{eq:action-3}\mathcal T_\times(\Phi)(X) &= [V^*J(\Phi)^\Gamma V] \odot X^\top,    
\end{align}
where $V : \mathbb C^d \to \mathbb C^d \otimes \mathbb C^d$ is the \emph{copy isometry} defined by its action on the standard basis $V e_i = e_i \otimes e_i$ and $\Gamma$ denotes the partial transposition: $A^\Gamma = [\operatorname{id} \otimes \operatorname{transp}](A)$. 
\end{proposition}
\begin{proof}
The first step of the proof is to consider the action of the twirling maps on the Choi-Jamio{\l}kowski matrices, which are represented graphically in Figure \ref{fig:twirl-choi} for the three cases.

The next step is to compute the averages of the three diagrams from Figure \ref{fig:twirl-choi} with respect to the uniform distributions over the random phase vectors $u,v$ (see Definition \ref{def:rv-u}). The averaging can be done graphically, by applying twice Theorem \ref{theorem:E-u}, once for $\mathbb E_u$ and once for $\mathbb E_v$, see also Remark \ref{rem:multiple-integrals-u}. The resulting diagrams are depicted in Figure \ref{fig:twirl-final}, from which Eqs.~\eqref{eq:action-1},\eqref{eq:action-2},\eqref{eq:action-3} can be immediately obtained. 
\end{proof}

\begin{remark}
The same result holds when replacing the Haar random diagonal unitary matrices $U,V$ with independent random diagonal sign matrices $S,T$. Indeed, Theorems \ref{theorem:E-u} and \ref{theorem:E-s} give identical formulas in the case of a single random matrix/vector ($n=1$). 
\end{remark}

\section{Conclusions and future directions}
We have developed a graphical calculus to compute expectations of tensor network diagrams consisting of random vectors with uniform, i.i.d.~components on the unit circle in the complex plane. While doing so, the natural partial order structure on the set of uniform block permutations was critically exploited. A near replica of this approach was implemented for calculating averages of diagrams consisting of random vectors with complex phases replaced by uniform, i.i.d.~random signs, where the poset of even partitions provided the necessary combinatorial structure. Applications of these results were presented in the analysis of several families of bipartite matrices with special local diagonal unitary/orthogonal invariance property. Notably, the notion of triplewise complete positivity was introduced to study the separability problem for these matrices. Finally, the twirling of linear maps between matrix algebras by independent diagonal unitary matrices provided an apt exhibit of the utility of our results.

From a combinatorial perspective, one could further investigate whether there exists some sort of a relationship between certain kinds of probability spaces and combinatorial structures. More precisely, can we consider averaging of a diagram containing random vectors from a different probability space, in order to obtain more exotic combinatorial objects such as \emph{non-crossing} partitions or \emph{non-crossing} uniform block permutations? 

The most obvious extension of our work would be to look for more applications of the graphical calculi, especially in cases where the involved number of random vectors is large $(\geq 3)$ and the computations become too cumbersome to do by brute force algebra. Decoupling by random diagonal unitaries and random diagonal unitary $t$-designs present compelling opportunities in this direction \cite{nakata2013diagonal, nakata2014diagonal, nakata2017unitary, nakata2017decoupling}. One could also look at the generalization of the notion of local diagonal unitary/orthogonal invariance in a multipartite setting. Separability analysis of the relevant matrices in this context might lead to more general notions of complete positivity of sets of tensors.

Even in the bipartite setting, the situation is far from being settled, and there are several unexplored avenues. For instance, even though some elementary results on triplewise completely positive matrices are presented in Appendix~\ref{sec:app-tcp}, one would like to ask for a more complete characterization of these matrix triples, so as to gain insights on the separability properties of the associated LDOI matrices. Development of sufficient conditions to guarantee triplewise complete positivity of a given matrix triple would certainly be desirable. One could also study explicit decompositions of given matrix triples, and ask questions regarding the existence and uniqueness of such decompositions. Significant work along some of these directions has been performed in \cite{singh2020diagonal,Singh2020entanglement,singh2020ppt2}.

\bigskip

\noindent\textit{Acknowledgements.} The work of SS was partially supported by an INSPIRE Scholarship for Higher Education by the Department of
Science and Technology, Government of India. We would like to thank Jeremy Levick and Mizanur Rahaman for engaging discussions regarding the potential applications of the graphical calculus to quantum information theory. 

\bigskip

\bibliographystyle{alpha}
\bibliography{references}

\appendix
\section{Application of Theorem~\ref{theorem:E-u} for the n=3 case}\label{sec:app-n-3}

If a diagram $\mathcal{D}$ has 3 $u$-boxes and 3 $\overbar u$-boxes, distributed according to Definition~\ref{def:rv-u}, then the number of diagrams in the expectation equals the number of uniform block permutations of order 3, i.e., $\vert \mathcal{UBP}_3 \vert = 16$, and we obtain the diagrams in Figure \ref{fig:E-u-3}. All the diagrams in the first bracket have coefficients $\Cf_\mathcal{U} = 1$, and correspond to permutations of the form ${\tiny\left(\!\!\begin{tabular}{c|c|c}
        1 & 2 & 3 \\
        $\ast$ & $\ast$ & $\ast$
    \end{tabular}\!\!\right)}$. The second bracket consists of diagrams which correspond to the UBPs of type $\lambda = 1^1 2^1$, each having a coefficient $\Cf_\mathcal{U} = -1$. Finally, the coarsest UBP of order 3 gives the last diagram with coefficient $\Cf_\mathcal{U} = +4$.

\begin{figure}[htb!]
    \centering
    \includegraphics[width=16cm]{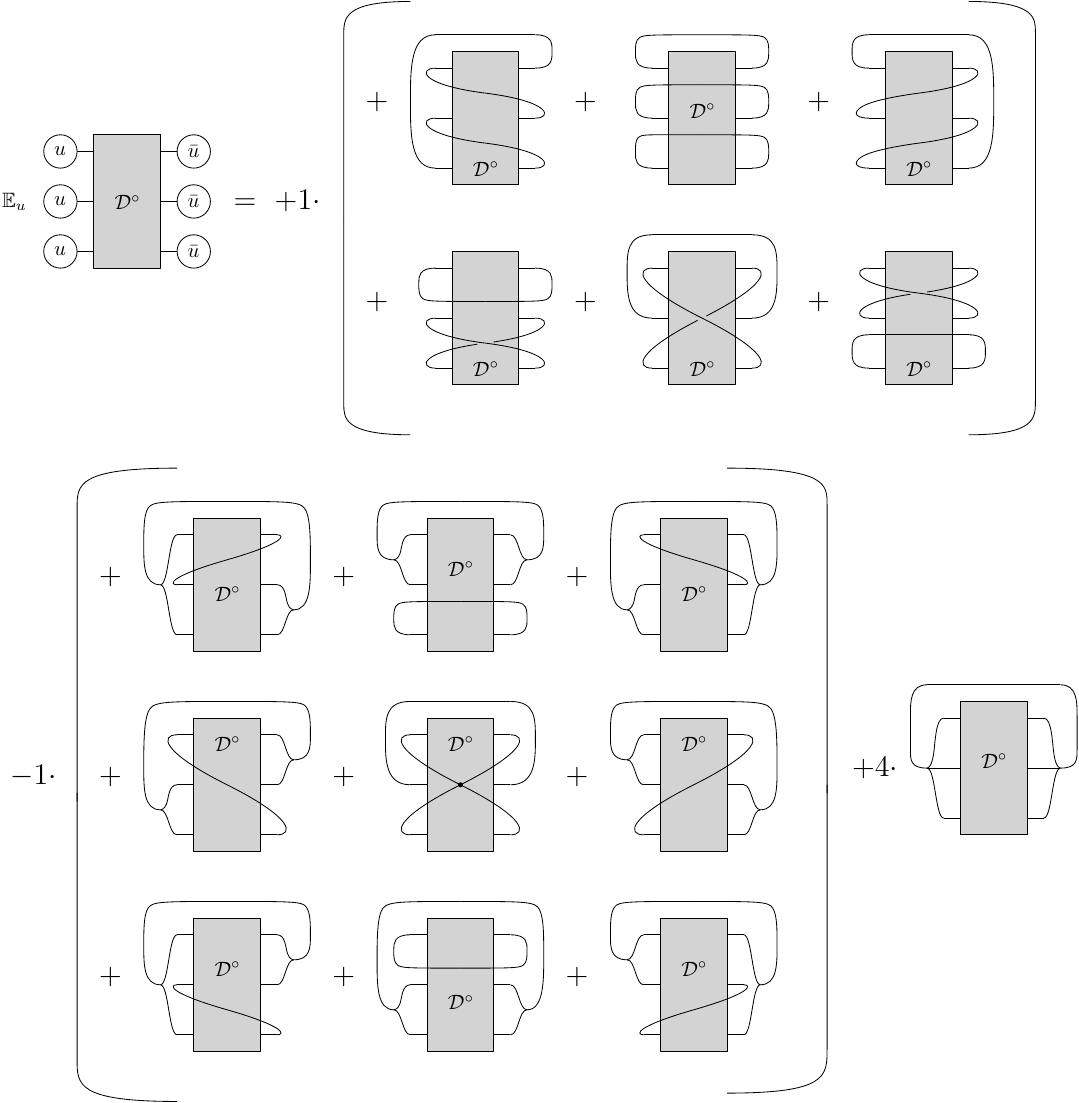}
    \caption{Applying Theorem \ref{theorem:E-u} in the case $n=3$. The result is a weighted sum of 16 diagrams.}
    \label{fig:E-u-3}
\end{figure}

\section{Triplewise completely positive matrices}\label{sec:app-tcp}
In this appendix, we collect some elementary results on triplewise completely positive (TCP) matrices and highlight their connection with the properties of the corresponding class of separable bipartite LDOI matrices.
The properties we discuss are analogous to those of pairwise completely positive (PCP) matrices presented in \cite{johnston2019pairwise}. A more extensive analysis is reserved for a future work \cite{singh2020diagonal}. All matrix triples considered in this appendix are in $\mathcal{M}_d(\mathbb{C}) \times \mathcal{M}_d(\mathbb{C}) \times \mathcal{M}_d(\mathbb{C})$.

\begin{lemma}
For $x,y \geq 0$ and triplewise completely positive $(A_1,B_1,C_1)$ and $(A_2,B_2,C_2)$, the triple $(xA_1 + yA_2, xB_1 + yB_2, xC_1 + yC_2)$ is triplewise completely positive, i.e., the set of triplewise completely positive matrices forms a convex cone.
\end{lemma}

\begin{proof}
Let $V_1, W_1$ in $\mathcal{M}_{d,d_1}(\mathbb{C})$ and $V_2, W_2$ in $\mathcal{M}_{d,d_2}(\mathbb{C})$ form the TCP decomposition of $(A_1, B_1, C_1)$ and $(A_2, B_2, C_2)$ respectively, see Definition \ref{def:tcp}. Then it is easy to check that the block matrices $V = [ \, \sqrt{x} \, V_1 \, | \, \sqrt{y} \, V_2 \,], \, W = [  \, W_1 \, | \,  W_2 \,]$ in $\mathcal{M}_{d,d_1 + d_2}(\mathbb{C})$ form the TCP decomposition for the convex combination $(xA_1 + yA_2, xB_1 + yB_2, xC_1 + yC_2)$.
\end{proof}

\begin{lemma}
Let $(A,B,C)$ be triplewise completely positive, $D$ be diagonal and entrywise non-negative and $P$ be a permutation matrix in $\mathcal{M}_d(\mathbb C)$. Then both $(DAD^*,DBD^*,DCD^*)$ and $(PAP^*,PBP^*,PCP^*)$ are triplewise completely positive.
\end{lemma}

\begin{proof}
If $V,W$ form the TCP decompositive of $(A,B,C)$ as in Definition \ref{def:tcp}, then it is clear that the matrices $V_1 = \sqrt{D}\, V, \, W_1 = \sqrt{D}\, W$ and $V_2 = PV , \, W_2 = PW$ form the TCP decomposition of $(DAD^*,DBD^*,DCD^*)$ and $(PAP^*,PBP^*,PCP^*)$ respectively, where the square root is understood to be taken entrywise.
\end{proof}

Next, we note some necessary conditions that a matrix triple $(A,B,C)$ must satisfy in order to have a chance at being triplewise completely positive. In what follows, $\Vert . \Vert_1$ and $\Vert . \Vert_{\operatorname{Tr}}$ denote the entrywise and trace norm on $\mathcal{M}_d(\mathbb{C})$ respectively.

\begin{lemma} \label{lemma:tcp-properties}
Triplewise complete positivity of $(A,B,C)$ entails that
\begin{enumerate}
    \item both $(A,B)$ and $(A,C)$ are pairwise completely positive.
    \item $\operatorname{diag}(A) = \operatorname{diag}(B) = \operatorname{diag}(C)$.
    \item $A$ is entrywise non-negative and $B, C$ are positive semi-definite.
    \item $A_{ij}A_{ji} \geq \vert B_{ij} \vert^2 $ and $A_{ij}A_{ji} \geq \vert C_{ij} \vert^2$ for $i,j \in [d]$.
    \item $\Vert A \Vert_1 - \Vert A \Vert_{\operatorname{Tr}} \geq \Vert B \Vert_1 - \Vert B \Vert_{\operatorname{Tr}}$ and $\Vert A \Vert_1 - \Vert A \Vert_{\operatorname{Tr}} \geq \Vert C \Vert_1 - \Vert C \Vert_{\operatorname{Tr}}$.    
\end{enumerate}
\end{lemma}

\begin{proof}
If $V,W$ form the TCP decomposition of $(A,B,C)$, then it is clear from Definition~\ref{def:tcp} that $V,W$ and $V,\overbar{W}$ form the PCP decompositions for the pairs $(A,B)$ and $(A,C)$ respectively. This proves part (1). All the other parts then follow from \cite[Theorem 3.2]{johnston2019pairwise}.
\end{proof}

It is worth noting that part (1) in the above Lemma is equivalent to saying that separability of an LDOI matrix with the associated triple $(A,B,C)$ implies separability of the CLDUI (resp.~LDUI) matrix with the associated  pair $(A,B)$ (resp.~$(A,C)$). Parts (3,4) are linked with the fact that separability of a bipartite matrix $X$ imply that both $X$ and its partial transpose $X^\Gamma$ are positive semi-definite, see Lemma~\ref{lemma:LDOI-psd-ppt}. Finally, part (5) is equivalent to the realignment criterion of separability for bipartite matrices \cite{Kai2002realignment, Rudolph2000realignment}.  Indeed, by defining the realignment map on $\mathcal{M}_d(\mathbb C) \otimes \mathcal{M}_d(\mathbb C)$ as $R(e_i e^*_j \otimes e_k e^*_l) = e_i e^*_k \otimes e_j e^*_l $, it is easy to see that for a separable CLDUI (resp. LDUI) matrix $X$ (resp. $Y$) with the associated pair $(A,B)$ (resp. $(A,C)$), the ``realigned" matrices $R(X)$ (resp. $R(Y)$) acquire the block structure given in Eq.~\eqref{eq:realign-LDUI/CLDUI} (see also Eq.~\eqref{eq:LDOI-block}), and the criteria $\Vert R(X) \Vert_{\operatorname{Tr}} \leq \operatorname{Tr}(X)$, $\Vert R(Y) \Vert_{\operatorname{Tr}} \leq \operatorname{Tr}(Y)$ translate into the conditions stated in part (5).

\begin{equation} \label{eq:realign-LDUI/CLDUI}
    R(X) =  A \oplus \, \, \bigoplus_{i<j} \left(   \begin{array}{cc}
        B_{ij} & 0  \\
        0 & B_{ji}
    \end{array}  \right) \qquad R(Y) = A \oplus \, \, \bigoplus_{i<j} \left(   \begin{array}{cc}
        0 & C_{ij}  \\
        C_{ji} & 0
    \end{array}  \right)
\end{equation}

We now strengthen the inequality in part (5) of Lemma~\ref{lemma:tcp-properties} for TCP matrices.

\begin{lemma}
If $(A,B,C)$ is triplewise completely positive, then the following inequality holds: 
\begin{equation}
    \Vert A \Vert_1 - \Vert A \Vert_{\operatorname{Tr}} \geq \sum_{i\neq j=1}^d \operatorname{max}\{ \vert B_{ij} \vert, \vert C_{ij} \vert \} \nonumber
\end{equation}
\end{lemma}
\begin{proof}
Since $(A,B,C)$ is TCP, the associated LDOI matrix $X$ is separable. Hence, the realigned matrix $R(X)$, with the block diagonal structure given in Eq.~\eqref{eq:realign-LDOI}, must satisfy $\Vert R(X) \Vert_{\operatorname{Tr}} \leq \operatorname{Tr}(X)$, i.e., $\Vert A \Vert_{\operatorname{Tr}} + \sum_{i<j} 2\operatorname{max}\{ \vert B_{ij} \vert, \vert C_{ij} \vert \} \leq \Vert A \Vert_1$, which is equivalent to the inequality stated in the lemma.  
\begin{equation} \label{eq:realign-LDOI}
    R(X) =  A \oplus \, \, \bigoplus_{i<j} \left(   \begin{array}{cc}
        B_{ij} & C_{ij}  \\
        C_{ji} & B_{ji}
    \end{array}  \right)
\end{equation}
\end{proof}

\begin{lemma}
For matrices $A,B,C$ in $\mathcal{M}_d(\mathbb{C})$, 
\begin{itemize}
    \item if $B$ (resp.~$C$) is diagonal, $(A,B,C)$ is triplewise completely positive if and only if $(A,C)$ (resp.~$(A,B)$) is pairwise completely positive.
    \item if $A$ is diagonal and entrywise non-negative, $(A,B,C)$ is triplewise completely positive if and only if $B$ and $C$ are also diagonal and equal to $A$.
\end{itemize}
\end{lemma}

\begin{proof}
For part (1), it suffices to note that if $B$ (resp.~$C$) is diagonal, the LDOI matix associated with the triple $(A,B,C)$ is equal to the LDUI (resp.~CLDUI) matrix associated with the pair $(A,C)$ (resp.~$(A,B)$), see Remark~\ref{remark:LDOI>LDUI}. 

For part (2), we observe that part (4) of Lemma~\ref{lemma:tcp-properties} already ensures that the off diagonal entries of $B$ and $C$ vanish if $(A,B,C)$ is TCP with $A$ being diagonal. Conversely, if $A=B=C$ are diagonal and entrywise non-negative, then the corresponding LDOI matrix is diagonal (see Eq.~\eqref{eq:LDOI-block}) and hence separable. 
\end{proof}

From the above Lemma, we can deduce that $(P, \operatorname{diag}(P), \operatorname{diag}(P))$ is TCP for all entrywise non-negative matrices $P$ in $\mathcal{M}_d(\mathbb{C})$. Combining this with the convexity of the set of TCP matrices yields the following corollary.

\begin{corollary} \label{corollary:ABC-P}
If $(A,B,C)$ is triplewise completely positive and $P$ is entrywise non-negative, then $(A+P, B+\operatorname{diag}(P), C+\operatorname{diag}(P))$ is triplewise completely positive.
\end{corollary}

For $A$ in $\mathcal{M}_d(\mathbb C)$, the equivalence of complete positivity of $A$ and pairwise complete positivity of $(A,A)$ was proven in \cite[Theorem 3.4]{johnston2019pairwise}. Lemma~\ref{lemma:AB-ABB} (stated and proved below) can be used to extend this equivalence to triplewise complete positivity of $(A,A,A)$.

\begin{lemma} \label{lemma:AB-ABB}
For $A,B$ in $\mathcal{M}_d(\mathbb C)$, $(A,B,B)$ is triplewise completely positive if and only if $(A,B)$ is pairwise completely positive.
\end{lemma}

\begin{proof}
The ``only if'' direction of the proof is evident from Lemma~\ref{lemma:tcp-properties}, part (1). Conversely, assume that $(A,B)$ is PCP with $V,W$ in $\mathcal{M}_{d,d'}(\mathbb C)$ forming its PCP decomposition, i.e.,
\begin{equation*}
    A = (V \odot \overbar{V}) (W \odot \overbar{W})^* \qquad B = (V \odot W) (V \odot W)^*
\end{equation*}

Now, define matrices $V', W'$ in $\mathcal{M}_d(\mathbb C)$ entrywise as follows: $V'_{ij} = V_{ij}\operatorname{phase}(W_{ij})$ and $W'_{ij} = \vert W_{ij} \vert$, where $\operatorname{phase}(W_{ij})$ are the complex phases of the entries of $W$, i.e.~$W_{ij} = \vert W_{ij} \vert \operatorname{phase}(W_{ij})$. It is easy to see that $V',W'$ also form a PCP decomposition of $(A,B)$. Moreover, since $W'$ is entrywise non-negative, $W' = \overbar{W'}$, which shows that $V',W'$ form a TCP decomposition of $(A,B,B)$ as well.
\end{proof}

\begin{corollary}
A matrix $A$ in $\mathcal{M}_d(\mathbb C)$ is completely positive if and only if $(A,A)$ is pairwise completely positive if and only if $(A,A,A)$ is triplewise completely positive.
\end{corollary}

\begin{corollary}
For $A,B$ in $\mathcal{M}_d(\mathbb C)$, if $A \geq B$ (entrywise) and $B$ is completely positive with $\operatorname{diag}(A) = \operatorname{diag}(B)$, then $(A,B,B)$ is triplewise completely positive.
\end{corollary}
\begin{proof}
Define $P = A-B$. Then, the assumptions of the corollary guarantee that $P$ is entrwise non-negative with $\operatorname{diag}(P)=0$ and $(B,B,B)$ is triplewise completely positive. An application of Corollary~\ref{corollary:ABC-P} then gives the desired result.
\end{proof}

The comparision matrix $M(X)$ of $X$ in $\mathcal{M}_d(\mathbb C)$ is defined entrywise as follows: \[
    M(X)_{ij} = 
\begin{cases}
    \vert X_{ij} \vert, \quad \text{if } i = j\\
    -\vert X_{ij} \vert, \quad \text{if } i \neq j 
\end{cases}
\]
In \cite[Theorem 4.4]{johnston2019pairwise}, a sufficient condition for pairwise complete positivty of a pair $(A,B)$ was derived using positive semi-definiteness of the comparision matrix $M(B)$. We use Lemma~\ref{lemma:AB-ABB} to extend this result for triplewise completely positive matrices. 

\begin{corollary}
Let $A,B,C=B$ be matrices in $\mathcal{M}_d(\mathbb C)$ such that the conditions mentioned in parts (2-4) of Lemma~\ref{lemma:tcp-properties} are satisfied. If in addition, $M(B)$ is positive semi-definite, then $(A,B,B)$ is triplewise completely positive.
\end{corollary}

Finally, we observe that since positivity under partial transposition is equivalent to separability for positive semi-definite matrices in $\mathcal{M}_2(\mathbb{C}) \otimes \mathcal{M}_2(\mathbb{C})$, the necessary conditions for triplewise complete positivity stated in Lemma~\ref{lemma:tcp-properties} are also sufficient when $d=2$. We state this more precisely in the following lemma.

\begin{lemma}
Let $A,B,C$ in $\mathcal{M}_2(\mathbb C)$ be such that the conditions mentioned in parts (2-4) of Lemma~\ref{lemma:tcp-properties} are satisfied. Then $(A,B,C)$ is triplewise completely positive.
\end{lemma}

\hrule
\end{document}